%% file: main.tex
\documentclass[acmsmall,screen,nonacm]{acmart}
\input{preamble}

\newif\ifsubmission\submissionfalse

\begin{document}

\title{Analyzing Decoders for Quantum Error Correction}

\author{Abtin Molavi}
\affiliation{\institution{University of Wisconsin-Madison}\country{USA}}
\email{amolavi@wisc.edu}

\author{Feras Saad}
\affiliation{\institution{Carnegie Mellon University}\country{USA}}
\email{fsaad@cmu.edu}

\author{Aws Albarghouthi}
\affiliation{\institution{University of Wisconsin-Madison}\country{USA}}
\email{aws@cs.wisc.edu}

\begin{abstract}
  Quantum error correction (\qec) enables reliable computation on noisy hardware by encoding logical information across many physical qubits and periodically measuring parities to detect errors.
  A \emph{decoder} is the classical algorithm that uses these measurements to infer which error most likely occurred, so that the system can correct it.
  The decoder's \emph{accuracy}---how rarely it makes the wrong guess---directly determines the scale of quantum computation that can be reliably executed.
  With a wealth of competing decoding algorithms, a \qec system designer needs reliable methods to evaluate them.
  Today, the dominant approach is to evaluate decoders using Monte Carlo simulation.
  However, simulation has several drawbacks such as requiring many samples to produce low variance estimates.

  In this work, we develop a new systematic analysis for evaluating decoders.
  We introduce a novel formal semantics of a core language for \qec programs
  that captures the de facto standard \stim circuit format, providing a principled theoretical foundation for the emerging space of fault-tolerant quantum systems design.
  Given a \qec program and a decoder, we can quantify both the decoder \emph{accuracy} and the decoder \emph{robustness} to drift in physical error rate.
  Our approach has two key components: \begin{enumerate*}[label=(\roman*)]
    \item a structured search over the space of possible errors; and
    \item a constrained polynomial optimization kernel.
  \end{enumerate*}
  A thorough empirical evaluation of our approach suggests that it can outperform simulation, especially in low error rate regimes, and that it can be deployed to quantify decoder robustness over an interval of physical error rates.
\end{abstract}
\maketitle

\section{Introduction}

Quantum computation promises to surpass classical methods in important domains, potentially
unlocking breakthroughs in materials science, chemistry, machine learning, and beyond. %
However, these practical applications often require billions or trillions of precise operations, while hardware quantum bits (qubits) are fragile and error-prone.
Fortunately, we can bridge the gap with quantum error correction (\qec).

\qec achieves fault-tolerance with redundancy, using several physical qubits to encode the state of each \emph{logical} qubit.
Repeated cycles of parity measurements detect errors without destroying the logical state.
In the absence of error, each cycle yields the same measurement outcomes. 
When a physical qubit is affected by an error between two measurement cycles, one or more of the measurements will flip (from 0 to 1 or 1 to 0).
The computational task of inferring precisely which error has occurred from the \emph{syndrome} of measurement outcomes is called \emph{decoding}.
Since there are multiple error patterns which can explain each syndrome, decoding is not a perfect process.
Instead, a decoding algorithm makes a probabilistic inference, aiming to find the \emph{most likely} error which can explain the syndrome.
 
When decoding fails and the inferred error is inequivalent to the actual error, the result is a \emph{logical error}.
The logical error rate, or \emph{decoder accuracy}, determines the reliability of an error-corrected quantum computer and the scale of computation which can be accurately executed.

\qec systems depend on efficient and accurate decoders to achieve acceptable logical error rates,
and a substantial research effort has been dedicated to the development of novel decoding algorithms.
The range of strategies that have been applied to the decoding problem include approaches based on 
Edmonds's classical blossom shrinking algorithm \cite{pymatching}, belief propagation in Bayesian networks \cite{relaybp,bp-lsd,bp-osd},
MaxSAT \cite{maxsat-decoder}, and neural networks \cite{alpha-qubit,alphaqubit2}.  

To make an informed decision from this array of options, \qec system designers need to evaluate different decoders, assessing the logical error rate they can expect from a decoder once deployed.
\textbf{The goal of our work is to improve the reliability of decoder evaluation.}

Typically, decoder accuracy is estimated via simulation. 
A \qec protocol is simulated with probabilistically injected errors according to a defined error model.
After collecting many simulation shots, the decoder accuracy can be empirically estimated by counting the fraction of shots where a logical error occurs.
This testing approach to decoder evaluation is exemplified by the \stim simulation framework \cite{Gidney2021stimfaststabilizer}, a popular library that supports high-performance, parallelized simulation and includes simple error models.
There are some important drawbacks to the testing strategy. 
\begin{itemize}
\item \textbf{Extremely rare events.} Simple sampling becomes less
effective as devices improve because the number of necessary samples scales inversely with the logical error rate.
A recent estimate suggests that using Shor's algorithm to factor a 2048-bit RSA integer would require a logical
error rate of $\approx10^{-15}$ \cite{gidney2025factor2048bitrsa}. Obtaining a reliable estimate of a quantity this small via simulation would take prohibitively many shots. 
\item \textbf{Uncertainty in modeling.} Models of the noise affecting a quantum computer are subject to uncertainty.
Calibration experiments observe drift in error rates over time, and transient burst events that temporarily spike error rates \cite{9951181,Google_Quantum_AI_and_Collaborators2024-fp,t1-fluctuations}.
Committing to concrete numerical values in simulation fails to capture the range of conditions that 
the decoder can reasonably be expected to operate in. 
\end{itemize}

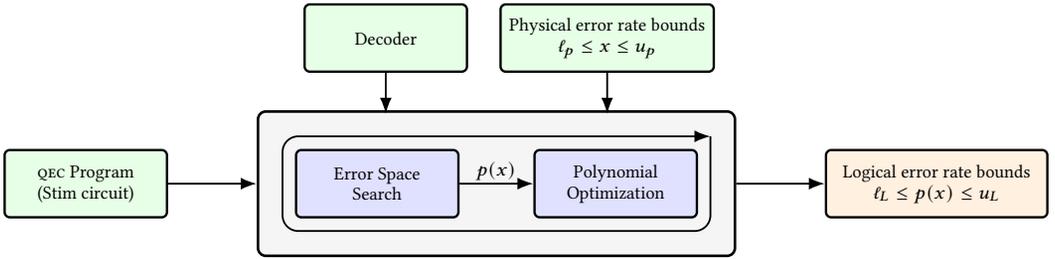
\begin{figure}[t]
\centering
\resizebox{\columnwidth}{!}{%
\begin{tikzpicture}[
  font=\scriptsize,
  arr/.style={-Latex, line width=0.7pt},
  lab/.style={midway, fill=white, inner sep=1.2pt},
  stage/.style={
    draw=black, line width=0.9pt, rounded corners=2pt,
    minimum height=9mm, minimum width=22mm,
    align=center, fill=blue!12
  },
  io/.style={
    draw=black, line width=0.9pt, rounded corners=2pt,
    minimum height=9mm, minimum width=22mm,
    align=center, fill=green!10
  },
  resultbox/.style={
    draw=black, line width=0.9pt, rounded corners=2pt,
    minimum height=9mm, minimum width=30mm,
    align=center, fill=orange!12
  },
  pipe/.style={
    draw=black, line width=1.0pt, rounded corners=3pt,
    inner sep=5mm, fill=gray!8
  }
]

% --- stages
\node[stage] (search) {Error Space\\Search};
\node[stage, right=10mm of search] (opt) {Polynomial\\Optimization};

% --- pipeline container BEHIND the stages
\begin{pgfonlayer}{background}
  \node[pipe, fit=(search) (opt)] (pipeline) {};
\end{pgfonlayer}

% --- p(x) arrow
\draw[arr] (search) -- node[lab,fill=none,above] {$p(x)$} (opt);
% --- perimeter feedback arrow: pipeline NE -> SE (wraps around stages)
\coordinate (pNE) at ($(pipeline.north east)+(-3.5mm,-3.5mm)$);
\coordinate (pNW) at ($(pipeline.north west)+( 3.5mm,-3.5mm)$);
\coordinate (pSW) at ($(pipeline.south west)+( 3.5mm, 3.5mm)$);
\coordinate (pSE) at ($(pipeline.south east)+(-3.5mm, 3.5mm)$);
\coordinate (pME) at ($(pipeline.north east)+(-3.5mm, -4.75mm)$);

\draw[arr, rounded corners=6pt]
  (pNE) -- (pSE) -- (pSW)-- (pNW) -- (pNE);
% --- input: Stim circuit (boxed) on the left (tighter)
\node[io, left=12mm of pipeline.west, yshift=0mm] (stim) {\qec Program \\ (Stim circuit)};
\draw[arr] (stim.east) -- ($(pipeline.west)$);

% --- input: Decoder (boxed) below (not wide)
\node[io, above=5mm of pipeline, xshift=-15mm] (dec) {Decoder};
\draw[arr] (dec.south) -- ($(pipeline.north)+(-15mm,-0.5mm)$);
% --- bounds note near optimize
\node[io, above=5mm of pipeline, xshift=15mm] (bounds) {Physical error rate bounds \\ $\ell_p \leq x \leq u_p$};
\draw[arr] (bounds.south) -- ($(pipeline.north)+(15mm,-0.5mm)$);
% --- output: boxed on the right (tighter)
\node[resultbox, right=12mm of pipeline.east] (res)
  {Logical error rate bounds \\
   $\ell_L \leq p(x) \leq u_L$};
\draw[arr] (pipeline.east) -- (res.west);

\end{tikzpicture}%
}
\caption{An overview of our approach.}
\label{fig:overview}
\end{figure}

\paragraph{Enumerative decoder analysis} To address these shortcomings, we propose a new approach to decoder evaluation, summarized in \zcref{fig:overview}.
Much like \emph{importance sampling}, we decouple the probability of an event from its priority in search.
We systematically enumerate events in the error space, visiting each event at most once and analytically computing its probability of occurrence. 
With this approach, we avoid the problem of wasting shots on high-probability events.  
For example, given realistic physical error rates, the event corresponding to ``no error'' is likely to dominate sampled outcomes.
With any reasonable decoder design, a logical error will not occur in the absence of physical error, so we gain little insight about the difference between decoding strategies from these samples.
For cases where the search space is too large to explore effectively with enumeration, we also present a hybrid approach which uses limited random sampling to estimate the contribution of the unexplored space to the logical error rate.

\paragraph{Quantifying decoder robustness} Our approach also enables analysis of decoder \emph{robustness}: the sensitivity of logical error rate
with respect to perturbations of the physical error rates.
The key insight is that the logical error rate can be represented \emph{symbolically} as a polynomial in the variables which represent physical errors.  
Computing robustness reduces to constrained optimization of the error polynomial.
We solve this constrained optimization problem exactly for the case where each physical error rate is constrained to an interval.
Practical solving is powered by a technique we call \emph{partial derivative pruning}, which can dramatically reduce the size of the search space in real problem instances.
Partial derivative pruning uses an over-approximation to analyze the sign of the partial derivative with respect to a particular variable across the entire search space.
If the sign can be determined, then we can choose to fix the variable to one of the boundaries, cutting the search space in half.

\paragraph{A formally specified input language} To specify the search space for the decoder accuracy problem, we define the notion of a \qec program, which represents the execution of a \qec protocol under a given noise model. 
Our \qec program abstraction formalizes the core of the \stim circuit language which is the input to the \stim simulation framework.
\stim circuits are emerging as a leading compilation target for higher-level descriptions of fault-tolerant quantum computation.
We provide solid foundations for the fault-tolerant quantum software stack with the first formal semantics. 
We also introduce an extension for defining robustness problems called \emph{symbolic} \qec programs. Each symbolic \qec program represents the execution of a \qec protocol under a parametric family of error models to verify decoder robustness against.

\paragraph{Evaluation} We evaluate our approach by estimating the accuracy and robustness of three state-of-the-art decoders on a leading class of \qec codes: rotated surface codes \cite{PhysRevA.86.032324}.
For the Accuracy problem, we compare against \stim simulation and find that our approach can find more precise estimates, especially in low error rate regimes.
For example, on a distance 5 surface code problem instance, our method is able to determine the logical error rate of a decoder within 0.1\% accuracy in about $3\cdot10^6$ shots, whereas sampling requires 100 times as many to confidently estimate the value within an order of magnitude.
\paragraph{Contributions}
In summary, our contributions are the following:

\begin{itemize}
    \item A formal syntax and semantics for \qec programs based on the \stim circuit format
    \item A definition of the decoder logical error rate estimation and robustness problems for \qec programs and a characterization in terms of polynomial optimization
    \item An algorithm for solving the logical error rate estimation and robustness problems based on systematic enumeration of the error space, including an extension which augments enumeration with some limited sampling
    \item A procedure for polynomial optimization over hyperrectangles that enables practical instantiation of our algorithm for a natural class of robustness problems
    \item An empirical evaluation that demonstrates the effectiveness of our algorithm as compared to simulation-based methods
\end{itemize}

\section{Quantum Computing and \qec Fundamentals}
In this section, we summarize the basic necessary background on quantum computing and quantum error correction. 

\paragraph{Qubits} The fundamental unit of quantum computing is called the quantum bit or \emph{qubit}.
Two possible states of a qubit are $\ket{0}$ and $\ket{1}$ which are analogous to the 0 and 1 of a classical bit. 
Unlike classical bits, any complex linear combination $\alpha\ket{0} + \beta\ket{1}$ of these two states is also a possible qubit state, subject to the restriction that $|\alpha|^2 + |\beta|^2 = 1$.
In other words, the state of a qubit is a unit vector in a two-dimensional vector space over $\mathbb{C}$.
A computational basis \emph{measurement} of a qubit in the state $\alpha\ket{0} + \beta\ket{1}$ results in the state $\ket{0}$ with probability $|\alpha|^2$ and $\ket{1}$ with probability $|\beta|^2$.
Multi-qubit systems are a straightforward extension. 
For example, the state of a two-qubit system is an element in a 4-dimensional vector space with basis $\{\ket{00}, \ket{01}, \ket{10}, \ket{11}\}$.
In general, the state of $n$ qubits is a vector in a $2^n$-dimensional vector space.

\paragraph{Density matrices} A density matrix is a structure that allows us to represent a probabilistic ensemble of quantum states.
Such an ensemble may be introduced by a measurement of one qubit of a multi-qubit system where the outcome is unknown. 
For a pure quantum state $\ket{\psi}$ like those we have seen so far, the density matrix $\rho$ is given by $\rho = \ket{\psi}\bra{\psi}$ where $\bra{\psi} = \ket{\psi}^\dagger$, the conjugate transpose.
For an ensemble of quantum states $\ket{\psi_i}$, each with probability $p_i$, $\rho = \sum_i p_i\ket{\psi_i}\bra{\psi_i}.$

\paragraph{Gates and circuits} 
Qubit states are transformed by operations called quantum gates.
Quantum gates are norm-preserving linear maps on the state. 
That is, a quantum gate can be represented as a unitary matrix.
One example is the $\X$ gate, which generalizes the \textsc{not} from classical circuits, mapping $\ket{0}$ to $\ket{1}$ and vice versa.
Another is the $Z$ gate, which applies a relative phase, mapping $\alpha\ket{0} + \beta\ket{1}$ to $\alpha\ket{0} - \beta\ket{1}$.
The $X$ and $Z$ gates are elements of an important class of single qubit gates called \emph{Pauli} gates, which includes $Y = iZX$ as the final member.

Gates can also act on more than one qubit.
A common two-qubit gate is the \cx gate. 
The \cx gate is named for its action on the computational basis states as a ``controlled-X.” If the first argument is $\ket{1}$ the \cx applies an \X to the second argument. 
If it is in the $\ket{0}$ state, the gate has no effect. 
Another example that will be relevant is multi-qubit Pauli operators.
An $n$-qubit Pauli operator is the tensor product of $n$ single-qubit gates which are either Pauli gates or the identity.
Quantum \emph{circuits} are the composition of quantum gates and measurements.

\paragraph{Quantum error correction} Much like classical error-correcting codes, quantum error-correcting codes reduce error rate with redundant encoding.
The state of a single logical qubit is encoded in multiple physical qubits, and measurements on ancilla qubits are used to detect and correct errors without collapsing the logical state.

The simplest example is the three-qubit repetition code, which encodes the state $\alpha\ket{0} + \beta\ket{1}$ as $\alpha\ket{000} + \beta\ket{111}$.
In this code, we detect errors with the quantum analog of classical parity measurements on the first two and last two qubits.
In \qec, such parity measurements are called stabilizer measurements.
Since the only nontrivial coefficients are on basis states where the qubits are all 0 or 1, the outcome of both stabilizer measurements is 0 with probability 1 in the absence of error.
An outcome other than 00 for the two measurements indicates that some qubit has erroneously flipped, which we can model as an unintended \X gate.
The outcomes of stabilizer measurements are called \emph{syndromes} because they are used to infer the presence of errors.
The other key type of parity measurements in a \qec code are \emph{logical observables}, measurement operations on the state of the logical qubit. 
In the three-qubit repetition code, a computational basis measurement of the logical state is equivalent to a computational basis measurement of any physical state (since all qubits are collapsed to the same basis state upon measurement), so we can define the logical observable as any single qubit measurement.

The three-qubit repetition code protects against bit-flips (erroneous $X$ gates), but qubits can also be affected by phase-flips, or erroneous $Z$ gates. 
A full \qec code includes stabilizer measurements to detect phase flips as well.

\section{The Accuracy and Robustness Problems: An Overview}
\label{sec:overview}
In this section, we establish the problem setting for this work, culminating in the definitions of the two decoder evaluation problems we address: Accuracy (\zcref{def:ler}) and Robustness (\zcref{def:rob}).

\paragraph{\qec programs}
We begin by introducing \qec programs, which represent the execution of an error-correction protocol under a particular noise model. 
\qec programs are a simple extension of quantum circuits with probabilistic noise operations and annotations for logical observables and stabilizer measurements.
\zcref{fig:rep-code-circuit} shows a simple example.
In this program, we prepare the logical $\ket{0}$ state in the three-qubit repetition code, perform one round of syndrome measurements,  then measure the logical qubit in the computational basis.
This pattern of preparing a basis state, performing one or more rounds of error-correction, then measuring in that basis is called a \emph{memory experiment}.
Memory experiments are a fundamental test used to validate the effectiveness of a \qec procedure at suppressing errors.

\begin{figure}
\begin{subfigure}[b]{.5\linewidth}
\begin{tabular}{c}
\begin{lstlisting}
// Apply bit-flip noise to qubits (w.p 0.01)
#\statement{Xerr}#(0.01) q[0];
#\statement{Xerr}#(0.01) q[2];
#\statement{Xerr}#(0.01) q[4];

// Perform parity check measurements
cx q[0] q[1];
cx q[2] q[3];
cx q[2] q[1];
cx q[4] q[3];
#\keyword{measure}# m[0] <- q[1];
#\keyword{measure}# m[1] <- q[3];
#\defsynd# m[0];
#\defsynd# m[1];

// Observable measurements, collapse state
#\keyword{measure}# m[2] <- q[0];
#\keyword{measure}# m[3] <- q[2];
#\keyword{measure}# m[4] <- q[4];
#\defobs# m[2];
\end{lstlisting}
\end{tabular}
\caption{\qec program}  
\label{fig:rep-code-circuit}  
\end{subfigure}\hfill
\begin{subfigure}[b]{.5\linewidth}
\small
\centering
\setlength{\tabcolsep}{6pt}
\begin{adjustbox}{max width=\linewidth}
\begin{tabular}{llll}
    \toprule
    Error & Syndrome & Observable & Probability \\
    \midrule
     000 & 00 & 0 & $(0.99)^3$  \\
     111 & 00 & 1 & $(0.01)^3$  \\
     001 & 01 & 0 & $(0.99)^2(0.01)$  \\
     110 & 01 & 1 & $(0.01)^2(0.99)$  \\
     011 & 10 & 0 & $(0.99)(0.01)^2$  \\
     100 & 10 & 1 & $(0.01)(0.99)^2$  \\
     010 & 11 & 0 & $(0.99)(0.01)(0.99)$  \\
     101 & 11 & 1 & $(0.01)(0.99)(0.01)$  \\
    \bottomrule
\end{tabular}
\end{adjustbox}
\caption{Induced probability distribution}
\label{fig:rep-code-dist}
\end{subfigure}
\caption{A \qec program for a three-qubit repetition code memory experiment and its induced probability distribution.
}
\label{fig:rep-code}
\end{figure}

In the example program, we see the four statement types in a \qec program.
\begin{itemize}
    \item \textbf{Standard circuit operations}: gates (e.g., Lines 7-10) and measurement (e.g., Lines 17-19)
    \item \textbf{Error channels}: Lines 2-4 apply \X gates with probability 0.01, a 1\% bit-flip error rate. 
    \item \textbf{Syndrome declarations}: Lines 13 and 14 define syndrome bits. This is the information that the decoder takes as input.
            In general, syndrome bits are defined as the parity of a set of measurement outcomes; here they are the trivial parity over a singleton set.
    \item \textbf{Observable declarations}: Line 20 declares a logical observable, a measurement that we can perform on the state of the logical qubit.
    The output of the decoder is a predicted value for this bit. 
    Like syndrome declarations, the expression here can be the parity of one or more bits of classical data.  
            
\end{itemize}
The semantics of an error correction program is a joint distribution over values of the declared syndrome and observable bits. 
For example, the distribution defined by this program is shown in \zcref{fig:rep-code-dist}. 
In the error column, we represent the applied noise operations as an \emph{error bitstring}. A "1" in index $i$ means the $i$th probabilistic operation is applied.
Note that the syndrome and observable outcomes are deterministic functions of the error.
That is, the error-free execution of a \qec program yields deterministic measurement outcomes, and the probabilistic outcomes are entirely due to probabilistic noise operations.

\paragraph{Observables and decoding} 
In our example program, we (arbitrarily) choose qubit 0 to define the logical observable.
When qubit 0 is affected by an \X error, the result is the state $\alpha\ket{100} + \beta\ket{011}$.
The value of our chosen observable has flipped, but as long as we can detect that a flip has occurred, we can recover our desired measurement operator.
We simply apply a corresponding correction to the measurement, interpreting a result of "1" as "0" and vice versa. 

A \emph{decoder} uses syndrome measurements to predict when an error has occurred that flips a logical observable.
With the explicit probabilities of each syndrome-observable pair in \zcref{fig:rep-code-dist}, we can use the following simple decoding algorithm: 
for each syndrome, compute the probability of the two observable outcomes, and predict the more likely outcome.
For example, given syndrome "00", we predict "0" because $(0.99)^3 > (0.01)^3.$

An incorrect decoder prediction is called a logical error, and the goal of decoding is to minimize the rate of logical errors.
The strategy described in the previous paragraph, called \emph{maximum-likelihood} decoding, is optimal for minimizing logical error rate.
Unfortunately, maximum-likelihood decoding is intractable for large-scale programs;
the general problem of computing the maximum likelihood observable outcome for a given syndrome is \#P-hard \cite{decoding-hardness,Fischer2024hardnessresults}.
Thus, decoders use suboptimal algorithms to achieve acceptable logical error rates in practice.
To evaluate decoders, we are interested in estimating their logical error rate.

\begin{definition}[Accuracy Problem]
    \label{def:ler}
    Given a \qec program \prog with $n$ syndrome bits and $m$ logical observables and a decoder $\decoder : \{0,1\}^n \to \{0,1\}^m$,
    find the logical error rate:
    \[\err(P, d) \coloneq \Pr_{(\synd, \obs) \sim \prog}[\decoder(\synd) \neq \obs].\]  
\end{definition}

\paragraph{Robustness} The program in \zcref{fig:rep-code-circuit} hard-codes a bit-flip probability of 0.01, but modeling quantum devices precisely is challenging, and error rates tend to drift over time \cite{Google_Quantum_AI_and_Collaborators2024-fp,t1-fluctuations}.
To capture this uncertainty, we can define the \emph{symbolic} version of the program that replaces the three concrete probabilities with symbolic variables $x_1, x_2$, and $x_3$.
To represent valid probabilities, we require $0 \leq x_i \leq 1$, but we are typically interested in worst-case performance over a much narrower range of values.
If our prior model estimates the bit-flip rate at 0.01, we might evaluate over the constraining set $0.009 \leq x_i \leq 0.011$ modeling a 10\% uncertainty in physical error rates.
We can quantify the robustness of a decoder over a constraining set by finding the maximum possible logical error rate when the symbolic variables are assigned to concrete values.

\begin{definition}[Robustness Problem]
    \label{def:rob}
    Given a symbolic \qec program $\prog(\vctr{x})$, decoder \decoder, and constraining set \rect,
    find the worst-case logical error rate:
    \[\max_{\vctr{v} \in R}{(\err(\prog(\vctr{v}), \decoder)}).\]
\end{definition}
The Robustness problem is a generalization of the Accuracy problem, with Accuracy equivalent to Robustness with a trivial constraining set where $|R| = 1$.
In \zcref{sec:opt-over-rect}, we describe our procedure for solving the Robustness problem over non-trivial constraining sets. 
In particular, we focus on \emph{hyperrectangles}, which capture the case described above, where each physical error rate is constrained to some independently defined interval.

Since the number of possible syndromes grows exponentially in the size of the logical qubit,
it is not feasible to solve the Accuracy and Robustness problems exactly.
We focus on computing sound bounds that guarantee the quantity of interest lies within some interval.

\section{The Language of \qec Programs}
In this section, we formalize the language of \qec programs, which represent the noisy execution of an error-correction protocol. 
Our language of \qec programs is motivated by the \stim \cite{Gidney2021stimfaststabilizer} circuit simulator; it corresponds to a core subset of the \stim circuit language.
\subsection{Syntax} 
A \qec program is a straight-line sequence of statements generated by the grammar in \zcref{fig:qec-grammar}.
Statements can act on classical and quantum register variables.
The term $q$ denotes a quantum register address, while $\overline{q}$ denotes an ordered list of such addresses.
Analogously, $m$ denotes a classical register address, and $\overline{m}$ an ordered list.
Like a traditional quantum circuit language (e.g. \textsc{openqasm} 2), possible statements include the application of a gate, qubit initialization, and measurement. 
In the statement $U(\overline{q})$, $U$ is a $2^{|\overline{q}|}$ dimensional unitary.
The statement \keyword{reset} $q$ sets the register $q$ to $\ket{0}$ and \keyword{meas} $m \gets q$ measures $q$ in the computational basis and stores the resulting bit in the classical register address $m$.

\qec programs also include statements specific to error correction.
The first type of such statements is an \emph{error channel}.
An error channel applies noise by probabilistically choosing a gate from a set of possible noise operators (which includes the identity $I$).
Typically, error channels are parameterized by a noise strength $p$.
The simplest error channels are the $\statement{Xerr}(p)$ and $\statement{Zerr}(p)$ channels which apply the corresponding gate with probability $p$ and $I$ otherwise. 
In addition to these basic error channels, error models also often include \emph{depolarizing} channels. 
The single qubit depolarizing channel of strength $p$, \statement{Depolarize1}(p), applies a gate from the set $\{X, Y, Z\}$ with probability $p$ (chosen uniformly, so each gate has a probability of $p/3$) and $I$ with probability $(1-p)$.
The two-qubit depolarizing channel $\statement{Depolarize2}(p)$ selects uniformly from the 15 non-identity $2$-qubit Pauli operators with probability $p$ and applies the identity with probability $(1-p)$.

The second type of error-correction statement is \emph{syndrome} and \emph{observable declarations}, which define the parameters of the decoding problem.
Syndrome declarations define a bit of the error syndrome as equal to the parity of an input list of classical registers, while observable declarations define logical observables in the same fashion.
Evaluating a \qec program amounts to probabilistically resolving these declarations to concrete bits. 
We assume, without loss of generality, that measurement statements are in SSA form, so that syndrome and observable declarations can be pushed to the end of the program.

\begin{figure}
  \begin{align*}
      P \coloneq &\ S;D \\
      S \coloneq  &\ \keyword{skip} \mid U(\overline{q}) \mid \keyword{reset } q \mid \keyword{measure } m \gets q \mid \mathit{Err}(\overline{q}) \mid S_1 ; S_2\\
      D \coloneq &\ \keyword{declare-syndrome}(\overline{m}) \mid  \keyword{declare-observable}(\overline{m}) \mid D_1;D_2  \\
      Err \coloneq &\ \statement{Depolarize1}(p) \mid \statement{Depolarize2}(p) \mid \statement{Xerr}(p) \mid \statement{Zerr}(p) \mid \cdots   \\
      p &\in [0,1]
  \end{align*}
  \caption{Syntax of a \qec program}
  \label{fig:qec-grammar}
\end{figure}

\subsection{Semantics}
\label{sec:semantics}
We assign a small-step semantics to a \qec program prefix with no syndrome or observable declarations by integrating ideas from prior definitions of quantum \cite{qwhile} and probabilistic \cite{sem-prob-progs-kozen,sampson2014passert} programs.
Specifically, we define a program configuration $\langle S, \Sigma, \rho, \Gamma \rangle$
where 
\begin{itemize}
  \item $S$ is a \qec program.
  \item $\Sigma$ is a sequence of random draws. 
  We model probabilistic error channels as functions with a random draw argument $\sigma$.
  The argument $\sigma$ is akin to a pseudorandom seed chosen uniformly from a finite set. 
  The applied noise in a \qec program is determined by the draw sequence $\Sigma$. 

  \item $\rho$ is a density matrix representing the current quantum state of the entire register of quantum variables drawn from a finite set $\mathit{QVar}$.
  \item $\Gamma : \mathit{CVar} \to \{0,1\}$ is the state of the classical register variables drawn from a finite set $\mathit{CVar}$.
\end{itemize}

We define a transition relation between program configurations:
\[ \langle  S, \Sigma, \rho, \Gamma \rangle \to \langle  S',\Sigma', \rho', \Gamma \rangle \]
to represent one step of execution.
The inference rules defining this relation are shown in \zcref{fig:smallstep}.
The \ruleref{Reset} and \ruleref{Unitary} rules say that these operations have the appropriate effect on the underlying quantum state.
The notation $U_{\overline{q}}$ in the \ruleref{Unitary} rule refers to the lifting of $U$ to the global state, applying $U$ to the relevant qubits and acting as the identity on others.\footnote{We can write $U_{\overline{q}}$ explicitly as $\Pi_{\overline{q}}(U \otimes I_{\text{dom}(\rho) \setminus \overline{q}})\Pi^{-1}_{\overline{q}}$ where $\Pi_{\overline{q}}$ is a permutation that shifts $\overline{q}$ to the first positions.}
The projection operator $\ket{0}\bra{0}_q$ in the \ruleref{Reset} rule is an analogous lifting.
The \ruleref{Measure} rule updates the quantum state in accordance with the measurement and stores measurement outcome in the classical store $\Gamma$. 
In this rule, $P_{b}$ is one of the projections onto the basis states: $\ket{i}\bra{i}_q$ for $i = 0$ or $i=1$.
Note that this rule allows non-deterministic transition to the state resulting from any possible measurement outcome.
We will return to this point in discussing \emph{well-defined} programs.
Finally, the \ruleref{Err} rule shows how error channels consume a random draw to choose a unitary to apply. 
The choice function $\textit{Err}(\sigma)$ encodes the relative probabilities of different unitaries.
For example, the statement $\statement{Xerr}(0.01)$ is modeled by a function $\textit{Err} : \{1, \ldots ,100\} \to \{I, X\}$ such that $\textit{Err}(1) = X$ and $\textit{Err}(x) = I$ otherwise.

\begin{figure}
\begin{mathpar}
\Rule{SkipSeq}{ }{\langle  \keyword{skip}; S, \Sigma, \rho, \Gamma \rangle \to \langle S, \Sigma, \rho, \Gamma \rangle}
\and
\Rule{Unitary}{\text{dim}(U) = 2^{|\overline{q}|} \and \overline{q} \subseteq \text{dom}(\rho)}{{\langle  U(\overline{q}), \Sigma, \rho, \Gamma \rangle \to \langle \keyword{skip} , \Sigma, U_{\overline{q}} \rho U_{\overline{q}}^{\dag}, \Gamma \rangle}}
\and
\Rule{Reset}{q \in \text{dom}(\rho) }{{\langle  \keyword{reset } q , \Sigma, \rho, \Gamma \rangle \to \langle \keyword{skip} , \Sigma, \ket{0}\bra{0}_q\rho \ket{0}\bra{0}_q + \ket{0}\bra{1}_q\rho \ket{1}\bra{0}_q, \Gamma \rangle}}
\and 
\Rule{Measure}{q \in \text{dom}(\rho) \and \text{Tr}(P_{b} \rho) > 0}{ {\langle  \keyword{measure } m \gets q , \Sigma, \rho, \Gamma \rangle \to \langle \keyword{skip} , \Sigma, \frac{P_{b} \rho P_{b}}{\text{Tr}(P_{b} \rho)}, \Gamma[m \mapsto b] \rangle}}
\and 
\Rule{Err}{\Sigma = \sigma : \Sigma' \and U = \mathit{Err}(\sigma) }{\langle \mathit{Err}(\overline{q}), \Sigma, \rho, \Gamma \rangle  \to \langle \keyword{skip}, \Sigma', U_{\overline{q}} \rho U^{\dag}_{\overline{q}}, \Gamma \rangle  }
\end{mathpar}
\caption{Small-step semantics for \qec program prefixes without syndrome and observable declarations}
\label{fig:smallstep}
\end{figure}

Syndrome and observable declarations denote simple functions on a classical register that take the parity of their argument bits, returning a single bit.
The composition of these statements concatenates the results into a pair of bitstrings: a syndrome bitstring and an observable bitstring.
Thus, we define a denotational semantics $\denote{\cdot}_D : D \times (\textit{CVar} \to \{0,1\}) \to \{0,1\}^* \times \{0,1\}^* $ on $D$, the set of declarations, as shown in \zcref{fig:denote}.

\begin{figure}
  \begin{align*}
    \denote{\keyword{declare-syndrome}(\overline{m})}_D(\Gamma)
    &= \left(\bigoplus_{m \in \overline{m}} \Gamma(m), \varepsilon\right) \\
  \denote{\keyword{declare-observable}(\overline{m})}_D(\Gamma)
    &= \left(\varepsilon, \bigoplus_{m \in \overline{m}} \Gamma(m)\right) \\
  \denote{D_1;D_2}_D(\Gamma)
    &= (s_1 \cdot s_2, o_1 \cdot o_2) \quad \text{where } (s_i, o_i) \coloneq \denote{D_i}_D(\Gamma)
\end{align*}
 \caption{Denotational semantics for declarations. The operator ``$\cdot$" is string concatenation}
 \label{fig:denote}
\end{figure}

\paragraph{Full program semantics} The result of an execution of a \qec program is a pair of bitstrings representing syndrome and observables obtained after evaluating the program body.
We can express this as a big-step judgment (\zcref{fig:bigstep}) combining the semantics for syndrome and observable declarations with the transition rules for other statements from \zcref{fig:smallstep}.
Here, $\rightarrow^*$ is the transitive closure of $\rightarrow$.

\paragraph{Well-defined syndromes and observables} Non-deterministic measurement means that we could write a \qec program such that the \ruleref{FullProg} above allows evaluation to multiple different syndrome-observable pairs. 
However, this behavior does not occur in programs that correspond to correctly constructed \qec codes and experiments.
We say a program that has one possible final evaluation is well-defined (see \zcref{def:well-def-prog}) and restrict our study to well-defined programs.
The restriction to well-defined programs can be enforced statically using a technique called stabilizer tableau simulation \cite{PhysRevA.70.052328,Gottesman:1997zz}.

\begin{definition}
  A \qec program $P$ is well-defined if for every configuration $(\Sigma, \rho, \Gamma)$, there is a unique $(s,o)$
  such that $\langle P, \Sigma, \rho, \Gamma \rangle \Downarrow (s,o)$.
  \label{def:well-def-prog}
\end{definition}

\begin{figure}
  \begin{mathpar}
    \Rule{FullProg}{\langle S, \Sigma, \rho, \Gamma \rangle \to^* \langle \keyword{skip}, \Sigma', \rho', \Gamma' \rangle \and \denote{D}_D(\Gamma') = (s,o)}{\langle S; D, \Sigma, \rho, \Gamma \rangle  \Downarrow (s, o)}
  \end{mathpar}
  \caption{Semantics of a full \qec program}
  \label{fig:bigstep}
\end{figure}

\paragraph{Error bitstrings: identifying executions}
To systematically evaluate decoders, we are interested in isolating particular execution paths of a \qec program. 
For a well-defined program, this amounts to logging which gates are actually applied by error channels.
To this end, we define the notion of an \emph{error bitstring}.
For simplicity, we assume that all error channels are \emph{Bernoulli}. That is, they draw from a set $\{U, I\}$ consisting of a single non-identity gate which is applied with some probability $p$ (depolarizing channels can be decomposed into several Bernoulli channels).
Then, an execution is defined by a bitstring $e$ where $e_i = 1$ represents applying the non-identity gate at the $i$th error channel statement, and $e_i = 0$ represents applying the identity.
The left column of \zcref{fig:rep-code-dist} shows all the error bitstrings for the example program.

\paragraph{Probabilistic semantics}
We now connect the formal semantics of \qec programs to the informal notion from \zcref{sec:overview} that a \qec program defines a distribution over syndrome-observable pairs.
As a stepping stone, we define a program's distribution over error bitstrings. 
Let $\prog$ be a \qec program with error channel statements $E_1, \ldots,  E_n$.
Then we write $e_\Sigma$ to denote the bitstring which results from the draw sequence $\Sigma = \sigma_1, \ldots, \sigma_n.$
That is, $e_\Sigma = e_1, \ldots, e_n$ is defined by the rule 
\[
  e_i = \begin{cases}
      0 &\text{ if } E_i(\sigma_i) = I \\
      1 &\text{otherwise}
  \end{cases}
\]
We say $\Sigma$ \emph{realizes} $e$ if $e_\Sigma = e$. 
With this notation we can express the probability distribution over bitstrings defined by a \qec program $\prog$ and random draw sequence space $\drawseqset$.
The probability of a bitstring is simply the proportion of draw sequences that realize it:

\[f(e; \prog, \drawseqset) = \frac{|\{\Sigma \in \drawseqset  : e_\Sigma = e\}|}{|\drawseqset|}.\]

Since the syndrome and observable outcomes are a deterministic function of the error bitstring, we can extend our distribution over error bitstrings into a distribution over syndrome-observable pairs.
We first leverage our program semantics to define the function from bitstrings to syndrome-observable pairs explicitly.
Given an error bitstring $e$, $\syndfunc(e)$ and $\obsfunc(e)$ denote the syndrome and observable pair that results from the corresponding execution, as formalized in \zcref{def:syndfunc}.

\begin{definition}
  \label{def:syndfunc}
  Let $P$ be a \qec program with $n$ Bernoulli error channel statements, and $e$ be an error bitstring realized by draw sequence $\Sigma_e$.
  We write $\syndfunc(e) = s$ and $\obsfunc(e) = o$ to mean the syndrome and observable such that $\langle P, \Sigma_e, {\ket{0}\bra{0}}^{\otimes n}, \varnothing \rangle  \Downarrow (s, o)$.
\end{definition}

We have a natural probability distribution over syndrome-observable pairs which is the probability of a bitstring which results in the pair:
\[g((s,o); \prog, \drawseqset ) = \sum_{e \in B} \chi_{s,o}(e; \prog)f(e; \prog, \drawseqset), \]
where $B$ is the set of error bitstrings for $\prog$ and $\chi_{s,o}(e; \prog)$ is the indicator function which is 1 when $\syndfunc(e) = s$ and $\obsfunc(e) = o$ and 0 otherwise.

\paragraph{Symbolic \qec programs}
To construct instances of the Decoder Robustness problems, we also need a notion of symbolic \qec program.
In a symbolic program, the concrete probabilities $p$ are replaced by variables $x$.
We write $P(\overline{x})$ for a symbolic program to indicate that each symbolic program defines a function from a vector of probabilities to a \qec program. 
We will sometimes also write $P(\vctr{v})$ for a concrete program to emphasize the perspective as the evaluation of a symbolic program at a point. 

For example, the symbolic version of the repetition code program from \zcref{fig:rep-code-circuit} is shown in \zcref{fig:rep-code-symb}.
This program corresponds to the program-valued function $P(x_1, x_2, x_3)$ and the program  from \zcref{fig:rep-code-circuit} is the evaluation $P(0.01, 0.01, 0.01).$
\begin{figure}
\begin{center}
\begin{tabular}{c}
\begin{lstlisting}
#\statement{Xerr}#(x1) q[0];
#\statement{Xerr}#(x2) q[2];
#\statement{Xerr}#(x3) q[4];
cx q[0] q[1];
cx q[2] q[3];
cx q[2] q[1];
cx q[4] q[3];
#\keyword{measure}# m[0] <- q[1];
#\keyword{measure}# m[1] <- q[3];
#\defsynd# m[0];
#\defsynd# m[1];
#\keyword{measure}# m[2] <- q[0];
#\keyword{measure}# m[3] <- q[2];
#\keyword{measure}# m[4] <- q[4];
#\defobs# m[2];
    \end{lstlisting}
    \end{tabular}
    \caption{A symbolic \qec circuit for a three-qubit repetition code memory experiment.}
    \label{fig:rep-code-symb}
\end{center}
\end{figure}

\section{The Probability Polynomial Lens}
\label{sec:prob-poly}
We now introduce a class of polynomials we call error polynomials that serve as a useful language for analyzing and solving the Accuracy and Robustness problems.
As we will show, the Accuracy Problem reduces to the evaluation of an error polynomial, and the Robustness problem reduces to constrained optimization of error polynomials.

We begin with computing the probability of a particular error bitstring.
We note that we can write the probability $f(e; P, \drawseqset)$ of a bitstring as a simple expression of individual error channel probabilities.
Each error channel statement represents an independent random variable, so the probability of a bitstring is simply the product of the probabilities of the outcomes at each index.
For example, in our three-qubit repetition code program, we have three error channel statements with associated probabilities: $x_1, x_2$, and $x_3$. 
One possible error bitstring is ``010,'' which has probability given by the expression $(1-x_1)(x_2)(1-x_3)$. We refer to this expression as the \emph{error minterm} in analogy to minterms from Boolean algebra.

\begin{definition}
  Given an error bitstring $e = e_1\ldots e_n $, the error minterm $m_e$ is the product
  \[\prod_{e_i = 1} x_i \prod_{e_i = 0}(1-x_i) \]
  over symbolic variables $x_1, \ldots, x_n$.
\end{definition}

Since error bitstrings represent disjoint events, the probability of a \emph{set} of error bitstrings of length $n$ is the sum of the probabilities of each individual bitstring. 
Note that this sum is a polynomial in the variables $x_1, \ldots, x_n$.

\begin{definition}
  Given a set of error bitstrings \errorset, the error polynomial $p_\errorset$ is the sum
  \[\sum_{e \in L} m_e .\]
\end{definition}

\begin{example}
  Consider our running example with the three-qubit repetition code. 
  The maximum-likelihood decoder makes a logical error on the set of bitstrings of Hamming weight at least 2: $\{111, 110, 011, 101\}$.
  The error polynomial corresponding to this set is 
  \[p_L = x_1x_2x_3 + x_1x_2(1-x_3) + (1-x_1)x_2x_3 + x_1(1-x_2)x_3.\] 
  \label{ex:prob-poly}
\end{example}

Observe that the probability of logical error for a decoder $\decoder$ is the probability of a set of error bitstrings, so it can be described by an error polynomial.
Namely, it is the probability of the set of errors for which $\decoder$ makes a logical error.
In \zcref{ex:prob-poly}, we found the polynomial that expresses the logical error rate for the optimal decoder on our running example.
This fact allows us to describe each of our decoder evaluation problems in terms of the polynomial $p_L$.

First, the Accuracy problem is equivalent to evaluation of $p_L$ at a particular point.

\begin{theorem}[Acc. as polynomial evaluation]

    Given a \qec program $\prog(\vctr{v})$ and a decoder $\decoder$,
    \[\err(P, d) = p_L(\vctr{v}), \]
    where $L = \{e  : \decoder(\syndfunc(e)) \neq \obsfunc(e) \}$. 
    \label{thm:ler-is-poly-eval}
\end{theorem}

Second, the Robustness problem is equivalent to constrained maximization of $p_L$.

\begin{theorem}[Robustness as polynomial maximization]
    Given a symbolic \qec program $\prog(\vctr{x})$, decoder \decoder, and constraining set \rect,
    \[\max_{\vctr{v} \in R}{(\err(\prog(\vctr{v}), \decoder)}) = \max_{\vctr{v} \in R} p_L(\vctr{v}), \]
    where $L = \{e : \decoder(\syndfunc(e)) \neq \obsfunc(e) \}$.
    \label{thm-rob-is-poly-opt}
\end{theorem}

\paragraph{Bounding error polynomials}
Of course, the polynomial formulation of these problems does not circumvent their fundamental intractability. 
The polynomial $p_L$ has exponentially many terms, so even evaluating it at a point is not necessarily feasible. 
Therefore, we derive bounds on $p_L$ that we can compute without constructing the entire polynomial, using only terms from a set $\seen$ of explored error bitstrings.
Our bounds rely on the following theorem, which essentially says that the probability of a set $\errorset$ is lower-bounded by the probability of elements of $\errorset$ we have seen,
and upper-bounded by the worst-case assumption that all bitstrings not already confirmed to be in $\errorset^C$ are actually in $\errorset$.
\begin{theorem}
  Given any two sets of length $n$ bitstrings \errorset and \seen, we have the following inequalities over the set $[0,1]^n$:
  \[
    p_{\errorset \cap \seen} \leq p_{\errorset} \leq 1-p_{\seen \setminus \errorset}.
  \]
  \label{thm:bounds}
\end{theorem}

\begin{corollary}
Given any two sets of length $n$ bitstrings \errorset and \seen and a constraining set $\rect \subseteq [0,1]^n$, we have the following inequalities:
\[  \max_{\vctr{x} \in R}(p_{\errorset \cap \seen}(\vctr{x})) \leq  \max_{\vctr{x} \in R}( p_{\errorset}(\vctr{x})) \leq 1 - \min_{\vctr{x} \in R} (p_{\seen \setminus \errorset}(\vctr{x})).\]
\label{thm:rob-bounds}
\end{corollary}
\section{Our Estimation Algorithm}
Our main algorithm systematically explores the space of error bitstrings and leverages the results from \zcref{sec:prob-poly} to iteratively refine bounds on the solution to either the Accuracy or Robustness problem.
We show the general Robustness version in \zcref{alg:rob-bound}. 

The algorithm iterates over the set of error bitstrings until exhausting all of them, or it is interrupted (most programs define a space of error bitstrings far too large to enumerate entirely).
In \zcref{sec:iter-strat}, we discuss choices for the order of iteration.
Each time a new bitstring is enumerated, we add it to the set of seen bitstrings (line 5).
Then, we convert it to a syndrome and observable pair, run the decoder on the syndrome, and check for a logical error (line 6).
If the result is a logical error, we add the bitstring to the set \errorset (line 7).
Finally, new bounds are derived for the solution to the Robustness problem instance via constrained polynomial optimization (lines 8 and 9).

In the special case of the Accuracy problem, the set $\rect$ is a singleton and the optimization in lines 8 and 9 reduces to evaluating the polynomials at a single point.
For the general Robustness problem, we must solve an optimization problem over an infinite domain.
In \zcref{sec:opt-over-rect}, we describe our approach for exact optimization over infinite constraining sets defined in terms of upper and lower bounds for each variable.

\begin{algorithm}
\begin{algorithmic}[1]
  \Procedure{bound}{$\prog(\vctr{x})$ : symbolic \qec program, $\decoder$ : decoder, $\rect \subseteq [0,1]^n$}
\State Let $B$ be the set of error bitstrings for $P$
\State Initialize empty sets $\errorset$ and $\seen$ of error bitstrings
\For{$e \in B$} \Comment{See \zcref{sec:iter-strat}} 
  \State $\seen.add(e)$
  \If{$\decoder(\syndfunc(e)) \neq \obsfunc(e)$}
    \State $\errorset.add(e)$
  \EndIf
  \State \textit{lower} $\gets \max_{x \in R}(p_{\errorset}(\vctr{x}))$ \Comment{See \zcref{sec:opt-over-rect}}
  \State \textit{upper} $\gets 1 - \min_{x \in R} (p_{\seen \setminus \errorset}(\vctr{x}))$
\EndFor
\State \Return (\textit{lower}, \textit{upper})
\EndProcedure
\end{algorithmic}

\caption{Our algorithm for sound bounds on the decoder robustness}
\label{alg:rob-bound}
\end{algorithm}

\subsection{Optimization over Hyperrectangles}
\label{sec:opt-over-rect}
In this subsection, we will develop a procedure for maximizing polynomials over a particular type of constraining set, a \emph{hyperrectangle}.
This procedure will enable solving the Robustness problem for the case where we define an interval of possible values for each physical error rate. 
We present our approach as a sequence of refinements. First we show the problem is decidable in exponential time. 
Then, we apply pruning and simplification techniques to reduce the search space  without sacrificing exact optimization. 
Our efforts culminate in \zcref{alg:poly-opt}, our full polynomial optimization algorithm.

\begin{definition}
  A hyperrectangle \rect is a subset of $\mathbb{R}^n$ described by a pair of vectors $\ell, u \in \mathbb{R}^n$ such that $x \in \rect\iff \ell \leq x \leq u$.
\end{definition}

The first key observation is that, though the hyperrectangle contains infinitely many points, 
extrema lie on one of finitely many vertices: points where each variable is set to its upper or lower bound. 

\begin{theorem}
  Let $p$ be an error polynomial over $n$ variables and \rect be a hyperrectangle with bounds $(\ell, u)$. Then, $p$ achieves its extreme values over \rect at a vertex $x^*$ such that $x^*_i \in \{\ell_i, u_i\}$
  for all $i \in \{1, \ldots, n\}$.
  \label{thm:opt-at-corner}
\end{theorem}
Intuitively, extreme points lie on vertices because an error polynomial is a multilinear function,
and a linear function achieves its extrema on a bounded interval at the endpoints.
This theorem yields an immediate exponential time algorithm for the optimization problem via exhaustive search over the vertices. 
Given the general problem of optimizing a multilinear polynomial is \textsc{np}-hard, we do not expect a polynomial time algorithm to exist.
However, we can do better than brute-force enumeration in practice by applying additional pruning of the discrete search over vertices.
\paragraph{Partial derivative pruning}
To avoid evaluating an error polynomial at every vertex, we apply a strategy analogous to branch-and-bound 
or DPLL.
We iterate over each variable $x_i$ individually and attempt to prove that the optimal point must lie on one of the two planes $x_i = \ell_i$ and $x_i = u_i$, cutting the remaining search space in half. 
To obtain such a proof, we need to determine the sign of the partial derivative $\frac{\partial p}{\partial x_i}.$

\begin{example}
Suppose we wish to \emph{maximize} the error polynomial \[p = x_1(1-x_2) + (1-x_1)x_2\] over the constraining set $0.009 \leq x_i \leq 0.011$.
We have that $\frac{\partial p}{\partial x_1} = (1-x_2) - x_2$, which is positive whenever $x_2 < 0.5$,
which is true of our entire constraining set. 
Therefore, $p$ is maximized by setting $x_1 = 0.011$. 
\label{ex:two-var}
\end{example}

In \zcref{ex:two-var}, determining the sign of the partial derivative was trivial because it is a single variable function.
In general, this step amounts to solving another polynomial maximization problem, just with one fewer variable. 
We relax the problem by introducing the possibility that we are unable to determine the sign, and the variable is left unassigned.
We compute sound bounds by optimizing minterms individually, since for a polynomial $p$ written as a sum  $\sum_{i=1}^k m_i$, we know
\begin{align*}
  &\min p = \min \sum_{i=1}^k m_i \geq \sum_{i=1}^k \min m_i \qquad \text{and} &\max p = \max \sum_{i=1}^k m_i \leq \sum_{i=1}^k \max m_i. 
\end{align*}

Optimizing an individual minterm is trivial, regardless of the total number of variables, because each variable appears exactly once either as an $x_i$ factor or a $(1-x_i)$ factor.
We can certify the sign of the partial derivative if the entire interval $[\sum_{i=1}^k \min m_i, \sum_{i=1}^k \max m_i ]$ has the same sign.
\begin{example}
  Applying partial derivative pruning to the polynomial $p = x_1(1-x_2) + (1-x_1)x_2$ from \zcref{ex:two-var}, we have the following bounds on $\frac{\partial p}{\partial x_1}$: 
\begin{align*}
  & \min \frac{\partial p}{\partial x_1} \geq  \min (1-x_2) + \min (-x_2) = (1-0.011) + (-0.011) = 0.978,  \\
  & \max \frac{\partial p}{\partial x_1} \leq  \max (1-x_2) + \max (-x_2)   = (1-0.009) + (-0.009) = 0.982.
\end{align*}
Since the entire interval $[0.978, 0.982]$ is positive, we can conclude that $x_1 = 0.011$ maximizes $p$.
\label{ex:bound-terms-indv}
\end{example}

\begin{example}
As an example where we cannot determine the sign of the partial derivative, suppose we wish to maximize the polynomial
\[p = x_1(1-x_2)x_3 + (1-x_1)x_2(1-x_3).\]

We have $\frac{\partial p}{\partial x_1}  = (1-x_2)x_3 - x_2(1-x_3)$. Bounding terms individually:

\begin{align*}
 & \min \frac{\partial p}{\partial x_1} \geq  \min ((1-x_2)x_3) + \min (-x_2(1-x_3)) = -0.002, \\
& \max \frac{\partial p}{\partial x_1} \leq  \max ((1-x_2)x_3) + \max (-x_2(1-x_3)) = 0.002.
\end{align*}
The interval $[-0.002, 0.002]$ contains both positive and negative values, so we cannot determine the sign of $\frac{\partial p}{\partial x_1}$, and fail to prune the search space. 
\end{example}

\paragraph{Matching term simplification}
An important optimization to make partial derivative pruning more precise is the simplification of \emph{matching terms} in the partial derivative.
Given a particular target index $x_i$, note that two terms that differ only at $x_i$ contribute the same magnitude to $\frac{\partial p}{\partial x_i}$ with opposite signs, cancelling each other out.
Before individually maximizing each term, we simplify algebraically, removing matching terms to improve the likelihood that partial derivative pruning will succeed.

\begin{example}
This example demonstrates how matching term simplification improves precision.
Suppose we wish to maximize the polynomial
\[p = x_1(1-x_2)x_3 + (1-x_1)(1-x_2)x_3 + x_1x_2x_3.\]
First, we attempt to apply partial derivative pruning without simplification.
We compute
\[\frac{\partial p}{\partial x_1} = (1-x_2)x_3 + -(1-x_2)x_3 + x_2x_3. \]
Then, we optimize individual terms:
\begin{align*}
  & \min \frac{\partial p}{\partial x_1} \geq  \min ((1-x_2)x_3) + \min (-(1-x_2)x_3) +  \min x_2x_3 = -1.9\cdot10^{-3},  \\
  & \max \frac{\partial p}{\partial x_1} \leq  \max  ((1-x_2)x_3) + \max (-(1-x_2)x_3) + \max x_2x_3 = 2.1\cdot10^{-3}.
\end{align*}
The interval $[ -1.9\cdot10^{-3}, 2.1\cdot10^{-3}]$ contains both positive and negative values, so we cannot determine the sign of $\frac{\partial p}{\partial x_1}$, and fail to prune the search space. 
On the other hand, if we simplify to $\frac{\partial p}{\partial x_1} = x_2x_3$, then we see
\begin{align*}
  & \min \frac{\partial p}{\partial x_1} \geq  \min x_2x_3 = 8.1\cdot10^{-5} \qquad \text{and}  
  & \max \frac{\partial p}{\partial x_1} \leq  \max x_2x_3 = 1.21\cdot10^{-4}.
\end{align*}
We thus conclude fixing $x_1 = 0.011$ maximizes $p$.
\end{example}

Our full polynomial maximization algorithm is shown in \zcref{alg:poly-opt}.
We iterate over each variable, applying partial derivative pruning in an attempt to fix each to one of the two bounds.
Then, we perform an exhaustive search over the remaining search space.

\begin{algorithm}[t]
\begin{algorithmic}[1]
  \Procedure{maximize}{$p$ : error polynomial, $\ell$, $u$}
\State Let $n$ be the number of variables in $p$
\State Initialize a solution vector $s_1, \ldots, s_n$
\For{$i$ in \{1, \ldots, n\}}
\State Let $p'$ be the simplified partial derivative of $p$ with respect to $x_i$.
\State $(\textit{lower}, \textit{upper}) \gets \textsc{bound-terms-individually}(p', \ell, u)$ \Comment{Procedure from \zcref{ex:bound-terms-indv}}
\If{\textit{lower} > 0}
\State Substitute $u_i$ for $x_i$ in $p$
\State $s_i \gets u_i$
\ElsIf{\textit{upper} < 0}
\State Substitute $\ell_i$ for $x_i$ in $p$
\State $s_i \gets \ell_i$
\EndIf
\EndFor
\State Let $j_1, \ldots, j_k$ be indices of remaining free variables
\State Find $s_{j_1}, \ldots, s_{j_k}$ which maximize $p$ by exhaustive enumeration
\State \Return $s$
\EndProcedure
\end{algorithmic}
\caption{Branching optimization of an error polynomial}
\label{alg:poly-opt}
\end{algorithm}

\begin{theorem}
  \zcref{alg:poly-opt} returns a maximum point for the given error polynomial over the given hyperrectangle.
  \label{thm:alg-correct}
\end{theorem}
\subsection{Enumeration Strategies}
\label{sec:iter-strat}
The pseudocode in \zcref{alg:rob-bound} leaves the order of traversal of the error bitstring space unspecified. 
In this subsection, we describe the strategies that we consider. 
We start with our simple core technique of Hamming weight order traversal, then two orthogonal extensions.

\paragraph{Hamming weight order} Our core strategy is to enumerate error bitstrings from least to highest Hamming weight. 
The probability of an error bitstring decays exponentially in the number of errors, so most of the probability mass is concentrated on low-weight bitstrings.
Therefore, a lower weight bitstring provides more information.
If a low-weight bitstring causes a logical error, it improves lower bounds significantly. 
If it does not, our upper bounds improve significantly. 

\paragraph{Split-search} The \emph{split-search} extension uses parallelism to discover the first logical errors more quickly. 
Note that we can parallelize the enumeration of error bitstrings, which introduces a choice of how to partition the search space among different enumeration threads.
Given $k$ workers, the most basic strategy is to assign worker $i$ to bitstrings whose position in the enumeration order is equal to $i$ modulo $k$. 
Another option is to distribute the workers across different regions of the search space. 
Theoretically, for an error-correction code with distance $d$ and optimal decoder, all logical errors should result from error bitstrings with weight greater than or equal to $d/2$. 
Therefore, enumeration of errors below this weight can improve upper-bounds but is unlikely to improve lower bounds.
In the \emph{split-search} strategy, we designate half of the workers to start enumeration at weight 0 and half to start enumeration at weight $\lfloor d/2 \rfloor + 1,$ where $d$ is a provided ansatz of the code distance. 

\paragraph{Local search} The local search extension is based on the intuition that small perturbations of a logical error bitstring are likely to also result in a logical error.
The local search strategy is \emph{adaptive}. 
Each time a logical-error-inducing bitstring is discovered, the pre-planned  Hamming weight order enumeration is paused, while neighbors reached by applying perturbations are explored.
The goal is to improve lower bounds more rapidly under the hypothesis that these related bitstrings will also result in a logical error.
Once the neighbor set is exhausted, the Hamming weight order search resumes.  
Local search is parameterized by a set of local moves, which define which neighbor strings to consider.
In this work, we evaluate two local move sets (and their union). 

The first local move set is \textsc{flip}, which includes the flipping of any individual position in a bitstring.
The motivation for this move set is that logical errors are caused by clusters of physical errors, and isolated changes elsewhere are likely to result in a new logical error bitstring. 

The second local move set is \textsc{shift}, which includes all circular shifts of a bitstring.
The shift operator $\tau$ shifts all bits right by one position (modulo the bitstring length): 
$$\tau(e_1\ldots e_n) = e_ne_1\ldots e_{n-1}.$$
The local moves in the set \textsc{shift} are the $n-1$ offsets $\tau, \tau^2, \ldots, \tau^{n-1}.$ 
Here, the motivation is the translational symmetry of error correction codes. 
\qec programs typically contain many copies of the same circuit applied to different sets of qubits.
The goal of shifting a bitstring is to reach a bitstring that corresponds to a translation of the same physical errors from one copy to another.
If a set of physical errors causes a logical error in the first, it likely does the same in the other.

\section{Combining Sampling and Enumeration}
\label{sec:sampling-and-enumeration}

In this section, we describe a technique for integrating the strengths of a sampling-based approach into our enumeration algorithm for the Accuracy problem.
Enumeration alone can struggle with large-scale \qec programs that exhibit a long tail of low probability bitstrings.
As a result, the enumerative search may leave a large portion of the probability mass unexplored, so the quantity $p_{S \setminus L}(\vctr{v})$ is relatively large.

In \zcref{thm:bounds}, we get an unconditionally sound lower bound by making the maximally conservative assumption that no unseen error bitstring corresponds to a logical error.
Likewise, our upper bound is obtained by assuming they all do.
However, these bounds can be quite loose.
In solving the Accuracy problem, we can get a more accurate estimate of the contribution of still unenumerated errors to the logical error rate via sampling.
Sampling error bitstrings from the unseen set and checking whether they cause a logical error can yield an empirical estimate of the probability that an unseen error causes a logical error.

Suppose we are solving some Accuracy problem with error bitstrings of length $n$.
Let $S \subset \{0,1\}^n$  be a subset of enumerated error bitstrings and let $L \subset  \{0,1\}^n$ be error bitstrings which cause a logical error.
Then we can decompose the logical error rate into the contribution from elements of $S$ and those outside of $S$:
\begin{align*}
  \textrm{Pr}_{(\synd, \obs) \sim \prog}[\decoder(\synd) \neq \obs] = p_L(v) = p_{L \cap S}(v) + p_{L \cap S^c}(v)
\end{align*}
\paragraph{Estimating unenumerated contribution} The problem is that we cannot evaluate $p_{L \cap S^c}(v)$ directly because we do not know the elements of $L \cap S^c$. 
Instead, we apply statistical estimation to the value  $p_{L \cap S^c}(v)$.
To this end, we rewrite $p_{L \cap S^c}(v)$ as an expected value:
\[p_{L \cap S^c}(v) = \sum_{e \in S^c} \chi_L(e) m_e(v) = p_{S^c}(v) \cdot \mathbb{E}_{e \sim \prog \mid e \in S^c}[\chi_L(e)],\] 
where $\chi_L$ is the indicator function which is 1 on inputs in $L$ and 0 elsewhere.
Note that $p_{S^c}$ is possible to evaluate directly with access to $S$ as $1-p_S$.
Let $\theta = \mathbb{E}_{e \sim \prog \mid e \in S^c}[\chi_L(e)].$
As the notation suggests, we view $\theta$ as the parameter for a Bernoulli random variable representing the ``logical error occurs'' event.
We can substitute an empirical estimate $\hat{\Theta}$ for $\theta$ to get an estimated logical error rate:
\[p_L(v) \approx p_{L \cap S}(v) + (1-p_{S}(v))\hat{\Theta}.\]
To obtain such an empirical estimate, we take $N$ samples $s_1, \ldots, s_N$ from $S^c$ (via rejection sampling) and compute the sample mean $\frac1N \sum_{i=1}^{N} \chi_L(e_i)$.

\paragraph{Confidence intervals} So far, we have established a point estimate of the logical error rate. 
However, analogous to our sound upper and lower bounds from enumeration, we would like some sort of guarantee that the true logical error rate lies within some range.
Since we are taking random samples, this guarantee is now a \emph{probabilistic} guarantee, yielding a confidence interval. 
For a confidence parameter $\alpha$, we can use a standard approach to derive $1-\alpha$ confidence intervals from Chernoff bounds (e.g. \cite[Section 10.1]{Lattimore_Szepesvári_2020}), as formalized in \zcref{thm:conf}.

\begin{theorem}
  \label{thm:conf}
  Let $\theta \in [0,1]$ be an unknown parameter
  and
  $X_1, \dots, X_N$
  be i.i.d.~$\mathrm{Bernoulli}(\theta)$.
  Let $\hat{\Theta} \coloneqq \frac{1}{N}\sum_{i=1}^{N}X_i$ be the sample mean.
  On the event $0 < \hat\Theta < 1$,
  define the random endpoints $L_\alpha \in [0, \hat\Theta]$ and $U_\alpha \in [\hat\Theta, 1]$
  to be the unique solutions to
  \begin{align*}
    KL(\hat{\Theta} || U_\alpha) =  KL(\hat{\Theta} || L_\alpha) = \frac{1}{N}\ln\left({\frac{2}{\alpha}}\right),
  \end{align*}
  where $KL(p || q) \coloneq p\ln\left(\frac{p}{q}\right) + (1-p)\ln\left(\frac{1-p}{1-q}\right)$
  is the Kullback-Leibler divergence.
  The random interval
  \begin{equation*}
  [L, U] \coloneq \begin{cases}
    \left[0, 1 - (\alpha/2)^{1/N}\right] & (\hat\Theta = 0) \\
    [L_\alpha, U_\alpha] & (0 < \hat\Theta < 1) \\
    \left[(\alpha/2)^{1/N}, 1\right] & (\hat\Theta = 1)
  \end{cases}
  \end{equation*}
  is a $1-\alpha$ confidence interval for $\theta$, i.e.,
  \begin{equation*}
  \Pr_\theta\sett*{\theta \notin [L, U]} \leq \alpha.
  \end{equation*}
\end{theorem}

Note that this theorem also applies to the use of sampling to directly estimate the logical error rate. 
We will use it to derive confidence intervals for \stim sampling in our experimental evaluation.

\section{Implementation and Evaluation}
\label{sec:eval}
We implemented our estimation algorithm as a Python library. 
The library takes \stim circuits as input. 
It leverages the \stim functionality to compile circuits to detector error models, which are a list of Bernoulli error channels with their probabilities and impact on the syndrome and observable.
From a detector error model, we can extract $\syndfunc(e)$ and $\obsfunc(e)$ for an error bitstring $e$. Our implementation is parallelized, spawning $k-1$ enumeration workers when run on a $k$-core machine. 
Each worker is tasked to explore a different region of the error bitstring space, running the decoder on each enumerated bitstring and classifying them into those which result in logical errors and those which do not. 

\paragraph{Experimental setup}  Unless otherwise noted, the following experimental conditions
apply to all empirical evaluations. Both our search and \stim sampling are allotted 16 cores of an AMD EPYC™ 7763
2.45 GHz Processor and 32 GB of RAM accessed via a distributed research cluster.
We terminate upon reaching 24hrs of runtime or $10^9$ shots.

\paragraph{\qec program benchmarks} We applied our approach to memory experiments with a leading class of quantum error correcting code: rotated surface codes.
We construct several memory experiment programs by varying the code distance $d$ and the number of rounds of syndrome extraction $r$. 
Specifically, we consider $d \in \{3,5,7,9\}$ and $r \in \{1,3,5, 7,9 : r \leq d\}$.
We used the \textsc{si1000} noise model \cite{yoked} in all cases. 
The \textsc{si1000} noise model was developed at Google for realistic modeling of errors on superconducting devices ("\textsc{si1000}" is short for superconducting-inspired, 1000 nanosecond error-correction cycle). 
It applies depolarizing error channels before each gate, reset, and measurement operation.
The noise strength is controlled by a single parameter $p$, which sets the depolarizing probability for two-qubit gates.
The error rates of other operations are defined as fixed multiples of $p$.
We evaluate with $p=0.001$, which roughly corresponds to the current state-of-the-art in quantum hardware, along an order of magnitude in both directions: $p=0.01$ representing a high-noise environment, and $p=0.0001$ to represent improved hardware of the near future.

\paragraph{Decoders} We evaluated the three different decoders that represent the current state-of-the-art: \decname{pymatch} \cite{pymatching}, \decname{bp-osd} \cite{bp-osd}, and \decname{relay-bp} \cite{relaybp}.
\decname{pymatch} is based on a reduction of the decoding problem to minimum-weight perfect matching.
\decname{bp-osd} and \decname{relay-bp} are based on belief-propagation in Bayesian networks.

\input{fig-conv-rate}
\input{fig-bp-osd-vary-noise-strength}
\input{fig-pymatch-d7-r1}

\paragraph{Research Questions} We designed our  evaluation to answer the following research questions. 
\begin{itemize}
  \item[\textbf{(RQ1)}] How efficient are our estimation strategies as compared to random sampling for the Accuracy problem?
  \item[\textbf{(RQ2)}] What bounds can we certify for the general Robustness problem?
  \item[\textbf{(RQ3)}] What is the impact of varying our error bitstring enumeration strategy?
\end{itemize}
\subsection{(RQ1) Accuracy}
To address (RQ1), we compared our approach to \stim simulation, the most widely-used sampling-based approach to the Accuracy problem. 
We turn the point estimate of logical error rate returned by \stim into a confidence interval using the two-sided Chernoff bound method from \zcref{sec:sampling-and-enumeration}.
For both the \stim sampling baseline and our own probabilistic bounds from combining enumeration with sampling, we fix $\alpha = 0.01$, so there is a 99\% probability that the actual logical error rate is within the interval defined by the bounds.

\zcref{fig:conv-rate} summarizes the efficiency of different approaches for the \decname{relay-bp} decoder on the surface code circuits.
``Enumeration''  refers to our base algorithm (\zcref{alg:rob-bound}) and ``Enumeration \& Sampling'' uses the method from \zcref{sec:sampling-and-enumeration} to obtain probabilistic bounds with 10,000 random samples.
For each algorithm, we plot the number of shots required to estimate the logical error rate within an order of magnitude.
That is, we find the minimum number of shots such that $u/\ell \leq \sqrt{10}$, where $u$ and $\ell$ are the upper bound and lower bound provided by the algorithm (respectively).
We emphasize that this direct comparison is not entirely apples-to-apples, since \stim sampling and Enumeration \& Sampling yield merely probabilistic bounds, while Enumeration bounds are unconditionally sound.
An `x' indicates that the corresponding algorithm failed to reach the convergence criterion within the time and memory bounds. 
If no algorithm reached the condition, the benchmark is excluded entirely.
We split the data into plots by noise strength to highlight key patterns in relative and absolute performance.

\paragraph{Comparing Enumeration to Sampling}
Decreasing the noise strength has an opposite effect on random sampling as compared to our two enumeration-based approaches.
For sampling, reducing the noise strength drastically increases the shots to convergence.
This matches our intuition, as the physical error rate decreases, so does the logical error rate.
Thus, we are attempting to estimate a smaller quantity, necessitating more samples. 
On the other hand, enumeration performs \emph{better} at lower noise strength.
A lower noise strength has no effect on the number of error bitstrings we are able to explore with enumeration because the runtime of the post-hoc computation of bitstring probability does not depend on the probability itself.
The dominant effect is the concentration of more of the probability mass on a few, low Hamming weight strings. 

The opposite scaling results in a different relative efficiency story at different noise levels.
We see \stim is at least one order of magnitude more efficient for analyzing the \decname{relay-bp} decoder on each program at the high-noise strength,
while the enumeration-based approaches are more efficient at the lowest noise strength.
We focus on this relationship by fixing a single circuit in \zcref{fig:bp-osd-vary-noise-strength}.
Each of these plots shows the bounds on logical error rate as a function of time.
The decoder \decname{bp-osd} and program ($d=5, r=1$ surface code) are fixed, but the noise strength varies. 
We expect physical error rates to decrease with time as quantum hardware continues to improve, so these results suggest our approach is relevant to the future of quantum system design.

\paragraph{Comparing Enumeration \& Sampling to Enumeration} The data also support the claim that the ``Enumeration \& Sampling'' approach adds value beyond enumeration alone.
For example, \zcref{fig:conv-rate-high-noise} shows Enumeration alone only converges for 2 benchmarks at high noise strength, but Enumeration \& Sampling is able to leverage the strengths of sampling in this regime and converges for all 14 programs.
We highlight this effect by fixing a benchmark in \zcref{fig:pymatch-d7-r1}.

\begin{mybox}
\textbf{Summary:} Sampling converges in fewer shots than Enumeration at the highest noise level, but the opposite is true at the lowest.
Combining Enumeration with Sampling universally reduces the shots required, sometimes dramatically, at the cost of trading unconditional bounds for probabilistic guarantees.
\end{mybox}

\subsection{(RQ2) Robustness}
Next, we turn to evaluating the scalability of our approach for the Robustness problem. 
For each of our benchmark instances, we construct a Robustness instance with 10\% uncertainty.
That is, given a concrete program with Bernoulli channel parameters $\vctr{v}$, we constrain $\vctr{x}$ in the symbolic program to the hyperrectangle $0.9 \cdot v \leq x \leq 1.1\cdot v$.

\input{fig-rob-conv-rate}

\zcref{fig:rob-conv-rate} summarizes the results similarly to \zcref{fig:conv-rate}.
However, notice here that the grouping of the bars is by decoder, since we have only a single robustness algorithm to evaluate.
It should be unsurprising that the general Robustness problem is harder to solve than the degenerate case of the Accuracy problem.
In our experiments, the result is that our algorithm can only prove tight Robustness bounds on relatively small programs. 
Nevertheless, verification of decoder robustness on small instances can increase confidence in performance on larger ones.
Note that even these small programs induce constrained optimization problems over huge search spaces.
For example, the $d=3, r=3$ circuits contain 286 error channel statements, so the maximization problem is over 286 variables.
A set of $2^{286}$ possible variable assignments is far beyond the reach of brute-force enumeration, so solving this instance demonstrates the effectiveness of partial derivative pruning.

\paragraph{Relative decoder robustness}
Decoder robustness provides a new source of comparison between decoders.
In our evaluation, we find the \decname{relay-bp} decoder performs worse than the other two in the presence of error rate uncertainty.
Among the 6 programs in \zcref{fig:rob-conv-rate}, the \decname{relay-bp} decoder exhibits the largest gap between accuracy value and the worst-case robustness value at a difference of 28.6\%.
The gap for the \decname{bp-osd} and \decname{pymatch} decoders are 21.6\% and 21.7\%, respectively.
For one \qec program ($d=3, r=3$, $p=0.001$), we even see a different relative ordering of decoders when focusing on robustness instead of accuracy. 
The \decname{relay-bp} decoder is more accurate, but less robust than the \decname{bp-osd} decoder.

\input{fig-split-search}

\begin{mybox}
  \textbf{Summary:} Our algorithm is able to quantify robustness for several benchmarks, with the largest being the $d=3, r=3$ surface code circuit, demonstrating the feasibility of our approach for solving the Robustness problem over hyperrectangles.
  The results of our robustness analysis suggest that the \decname{relay-bp} decoder may be less robust than \decname{pymatch} and \decname{bp-osd}.
\end{mybox}
\subsection{(RQ3) Enumeration Strategies}
Finally, we study the effect of the choice of strategy for enumerating error bitstrings on the performance of our algorithm.
We focus on Accuracy problems to isolate enumeration from solving.

\paragraph{Split-search} We begin with an evaluation of the ``split-search'' strategy. 
Recall that this strategy divides a pool of enumeration workers in half, assigning one half to begin their search from errors of weight $\lfloor d/2\rfloor + 1$. 
We compare the split-search strategy to standard enumeration by Hamming weight on three representative example programs in \zcref{fig:split-search}.

Overall, we find that the ``split-search'' strategy and standard enumeration have early differences in bound tightness, but tend to converge as the number of shots increases.
In the limited shot regime, the better approach depends on the program. 
For the distance 5 example, the higher concentration of bitstrings which cause the decoder to make a logical error is outweighed by the lower probability of the visited bitstrings.
However, at larger code distances the ``jump'' forward into the high-weight error bitstring space is effective. 
In the distance 9 program, standard enumeration takes about $6\times$ as many shots to reach the first recorded lower bound from the split-search method because of effort spent enumerating low-weight error bitstrings which do not correspond to logical errors.

\paragraph{Local search} 
Next, we evaluate the impact of the local moves extension. 
We show a comparison to the baseline enumeration across the entire benchmark suite in \zcref{fig:ls-tables}.
For each local move set we record two summary statistics. 
The first column is the ``fraction improved,'' counting the number of benchmarks (program, decoder pairs) for which the local move strategy converged in fewer shots than the baseline. 
The second column captures the magnitude of difference. 
The shot ratio is defined as the number of shots required for the local move strategy to converge divided by the number of shots for the baseline to converge. 
A shot ratio less than 1 thus indicates better performance.
We take the geometric mean of shot ratio over all benchmarks.
The two tables separate the Enumeration bounds from the probabilistic Enumeration \& Sampling bounds.

For both types of bounds, we see that \textsc{shift} strategy is an improvement over the baseline.
However, the improvement in Enumeration bounds is more pronounced.
We see faster convergence on all but one benchmark for an average improvement of 27\%, as compared to faster convergence on only 68\% of benchmarks for Enumeration \& Sampling for an average improvement of 16\%.
On the other hand, the \textsc{flip} strategy is roughly on par with the baseline, improving on a little over half of the benchmarks, though with an average shot ratio greater than 1 indicating a larger gap in cases where the baseline performed better.
The combination of the two strategies yields intermediate results.

\input{fig-ls-tables}

\begin{mybox}
  \textbf{Summary:} 
  Overall, the effectiveness of alternative enumeration strategies is mixed.
  The split-search strategy induces different behavior in the low-shot regime, underperforming relative to standard Hamming weight order enumeration for low-distance benchmarks, and out-performing for high-distance benchmarks.
  In terms of local search, the \textsc{flip} strategy results in roughly equivalent performance to the baseline, while the \textsc{shift} strategy appears slightly more effective.
\end{mybox}
\section{Related Work}

\paragraph{Quantum programming languages}
As previously mentioned, our notion of \qec programs is a direct formalization of the \stim circuit format \cite{Gidney2021stimfaststabilizer}.
There is an extensive history of formally defined quantum programming languages, typically designed with program verification in mind.
Our language of \qec programs is closely related to the \keyword{qwhile} language for quantum programs \cite{qwhile}, which was introduced with a denotational semantics and relatively complete Hoare-style proof system.
Subsequent work has mechanized the semantics of \keyword{qwhile} in the Roqc \cite{CoqQ} and Isabelle/HOL \cite{qhl} proof assistants. 
Other notable examples of formalized quantum programming languages include \qwire \cite{qwire}, \sqir \cite{voqc}, and \qbricks \cite{qbricks}.
Our language is simpler than these examples in that it does not support unbounded looping or classical flow.
On the other hand, probabilistic error channels and syndrome and observable declarations are a unique feature of the error correction setting that we consider in our formal semantics.

\paragraph{Probabilistic programming languages}
Probabilistic programming languages are a rich area of study \cite{sampson2014passert,sem-prob-progs-kozen,dice2020,fairsquare,towner2019pldi,10.5555/3023476.3023503,problog}.
Like many other applications of probabilistic programming, the Accuracy and Robustness problems are probabilistic inference tasks over our programs.
A fundamental difference between our setting and the typical case is that we assume only black-box access to the decoder itself, since decoding algorithms are too complex to embed directly in a probabilistic programming language.
Therefore, we cannot directly apply symbolic techniques like weighted model counting \cite{dice2020,10.1145/3729334}.

\paragraph{Heuristic methods for decoder accuracy} 
A different approach to addressing the shot inefficiency of direct Monte Carlo sampling has been developed concurrently to this work.
\citet{beverland2025failfasttechniquesprobe} propose two main techniques. 
One is a multi-seeded \emph{splitting method} (building off of prior work \cite{Bravyi_2013,mayer2025rareeventsimulationquantum}).
The splitting method is a standard Monte Carlo algorithm for simulating rare events \cite{rareeventsimbook} based on multilevel Markov Chain Monte Carlo sampling. 
In this context, it is used to extrapolate the logical error rate at a low noise strength using a high-confidence estimate at a higher one.
The other technique is to apply curve-fitting to the \emph{failure spectrum}, the proportion of error bitstrings which cause a logical error at each Hamming weight.
Given a closed-form expression for the failure spectrum, they derive a logical error rate by summing over the contribution at each Hamming weight.
\citet{ye2026scalabletestingquantumerror} also develop a similar curve-fitting framework.
Our work distinguishes itself in the following ways.
\begin{enumerate}
  \item We make no assumptions on physical error model. The curve-fitting approaches can only natively model uniform probabilities of physical errors.
  \item We make no assumptions on logical error structure. Splitting methods assume the set of error bitstrings which cause a logical error is connected under single bitflips, and curve-fitting adds a modeling assumption in the choice of a smooth failure spectrum ansatz.
  \item We can formulate and solve the general Robustness problem.
\end{enumerate}

\paragraph{Analytic and symbolic \qec verification}
Handcrafted threshold existence proofs \cite{yoshida2026prooffinitethresholdunionfind,mwpm-threshold} have been applied to specific decoding algorithms and error-correction codes to show the existence of a finite \emph{threshold}, a physical error rate below which the logical error rate asymptotically tends to 0 as code distance increases. 
This is complementary to our black-box, automatic numerical analysis. 
Another line of work applies symbolic technique and SMT solving to verify \qec protocols \cite{Fang_2024,efficient-qec-verification-pldi-25}.
Here the focus is on the verification of a \qec program itself, assuming a perfect decoder, and checking for the existence of (or verifying the absence of) small sets of physical errors which cause a logical error.

\section{Conclusion} 
In this work, we presented a new pathway for formal reasoning about the behavior of decoders powered by a formal syntax and semantics for \qec programs and a reduction to polynomial optimization.
Our perspective enables the first definition of the Robustness problem for \qec decoders.
Extending our initial exploration of the Robustness problem is a promising direction for future work.
For example, it would be valuable to develop techniques for solving Robustness problems over constraining sets beyond hyperrectangles or to develop decoding algorithms with robustness guarantees in the spirit of robust training in machine learning \cite{dro}.

\section*{Data-Availability Statement}
Our paper is associated with a software artifact that includes our benchmarks, the implementation of our algorithm, and scripts for running experiments and generating plots. 
We will submit this artifact packaged in a Docker image for artifact evaluation.
A preliminary version is available at this URL: \url{https://anonymous.4open.science/r/analyzing-decoders-8034/}.
\bibliographystyle{ACM-Reference-Format}
\bibliography{references}

\ifsubmission\else
  \input{appendix}
\fi

\end{document}

%% file: preamble.tex
\usepackage[T1]{fontenc}
\usepackage{xspace}
\usepackage{listings}
\usepackage{xcolor}
\usepackage{amsmath}
\usepackage{zref-clever}
\usepackage{booktabs}
\usepackage{color}
\usepackage{braket}
\usepackage{mathtools}
\usepackage{tikz}
\usepackage{algorithm}
\usepackage{mathpartir}
\usepackage{langrules}
\usepackage{subcaption}
\usepackage[noend]{algpseudocode}
\usepackage{tikz}
\usepackage[inline]{enumitem}
\usepackage{tcolorbox}
\usepackage{adjustbox}
\usepackage{wrapfig}
\usetikzlibrary{arrows.meta, positioning, fit, calc}

\pgfdeclarelayer{background}
\pgfsetlayers{background,main}
\zcsetup{cap,abbrev}

\newif\ifcomments\commentstrue
\newcommand{\newcommenter}[3]{%
  \ifcomments
    \newcommand{#1}[1]{%
      \begingroup\small\sffamily\color{#2}%
        [#3: ##1]%
      \endgroup%
    }%
  \else
    \newcommand{#1}[1]{\ignorespaces}
  \fi
}
\definecolor{darkgreen}{rgb}{0,0.7,0}
\newcommenter{\abtin}{violet}{Abtin}
\newcommenter{\fsaad}{blue}{Feras}
\newcommand{\stim}{Stim\xspace}
\newcommand{\decname}[1]{\textsc{#1}\xspace}
\DeclareMathOperator{\err}{Err}
\DeclareMathOperator{\syndfunc}{synd}
\DeclareMathOperator{\obsfunc}{obs}
\DeclarePairedDelimiter{\sett}{\lbrace}{\rbrace}
\newcommand{\qec}{\textsc{qec}\xspace}
\newcommand{\X}{\ensuremath{X}\xspace}
\newcommand{\cx}{\textsc{cx}\xspace}
\newcommand{\prog}{\ensuremath{P}\xspace}
\newcommand{\defsynd}{\keyword{declare-syndrome}}
\newcommand{\defobs}{\keyword{declare-observable}}
\newcommand{\decoder}{\ensuremath{d}\xspace}
\newcommand{\synd}{\ensuremath{s}\xspace}
\newcommand{\obs}{\ensuremath{o}\xspace}
\newcommand{\vctr}[1]{\overline{#1}}
\newcommand{\errorset}{\ensuremath{L}\xspace}
\newcommand{\seen}{\ensuremath{S}\xspace}

\newcommand{\rect}{\ensuremath{R}\xspace}
\newcommand{\statement}[1]{\texttt{#1}}
\newcommand{\keyword}[1]{\textbf{\statement{#1}}}
\newcommand{\denote}[1]{\mathchoice%
  {\left\llbracket {#1} \right\rrbracket}
  {\llbracket {#1} \rrbracket}
  {\llbracket {#1} \rrbracket}
  {\llbracket {#1} \rrbracket}%
}
\newcommand{\sqir}{s\textsc{qir}\xspace}
\newcommand{\qwire}{\ensuremath{\mathcal{Q}\textsc{wire}}\xspace}
\newcommand{\qbricks}{\textsc{qbricks}\xspace}
\newcommand{\drawseqset}{\mathcal{D}}
\AddToHook{env/definition/begin}{%
   \zcsetup{countertype={theorem=definition}}}
\zcRefTypeSetup{definition}{Name-sg=Definition}
\AddToHook{env/example/begin}{%
   \zcsetup{countertype={theorem=example}}}
\zcRefTypeSetup{corollary}{Name-sg=Corollary}
\AddToHook{env/corollary/begin}{%
   \zcsetup{countertype={theorem=corollary}}}

\newenvironment{mybox}[1][gray!10]{
    \begin{tcolorbox}[
        left=0pt,
        right=0pt,
        top=0pt,
        bottom=0pt,
        colback=#1,
        colframe=#1,
        width=0.99\linewidth,
        boxsep=2pt,
        arc=0pt,outer arc=0pt,
    ]
}{
    \end{tcolorbox}
}
\zcsetup{cap, nameinlink=false, abbrev}

%% file: fig-conv-rate.tex
%!TEX root=./main.tex
\begin{figure}
  \captionsetup[subfigure]{skip=0pt}
  \begin{subfigure}[b]{0.8\textwidth}
    \includegraphics[width=\textwidth]{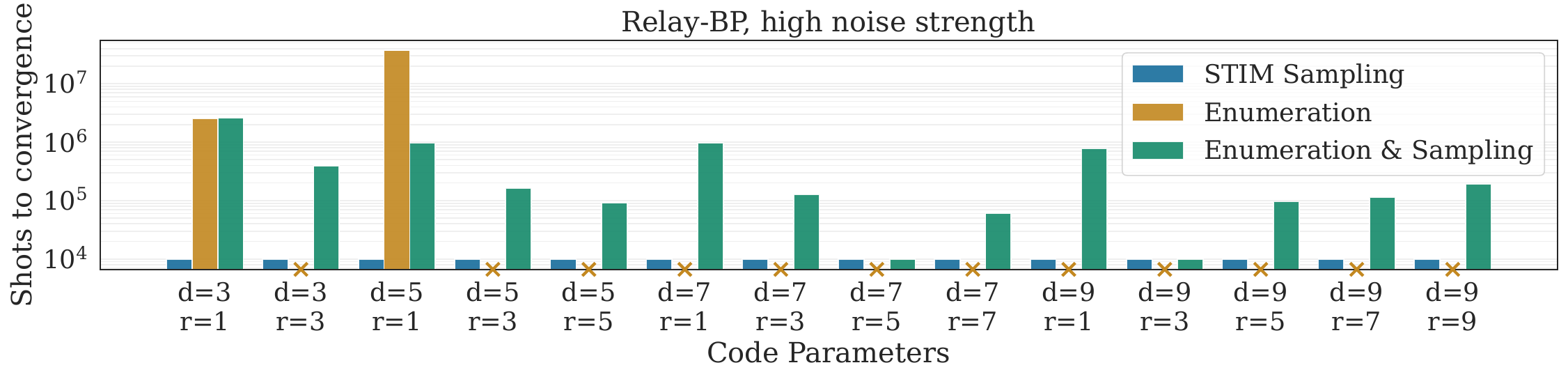}
    \captionsetup{belowskip=6pt}
    \caption{$p = 0.01$}
    \label{fig:conv-rate-high-noise}
  \end{subfigure}
    \begin{subfigure}[b]{0.7\textwidth}
      \includegraphics[width=\textwidth]{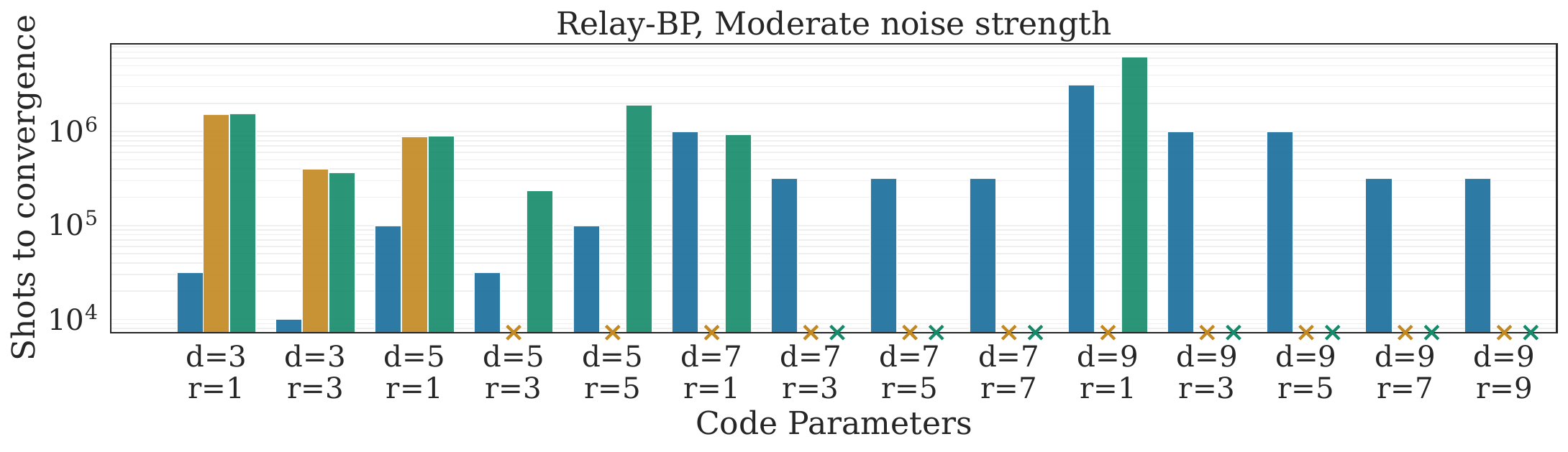}
      \captionsetup{belowskip=6pt}
      \caption{$p = 0.001$}
      \label{fig:conv-rate-mod-noise}
    \end{subfigure}
    \begin{subfigure}[b]{0.45\textwidth}
    \includegraphics[width=\textwidth]{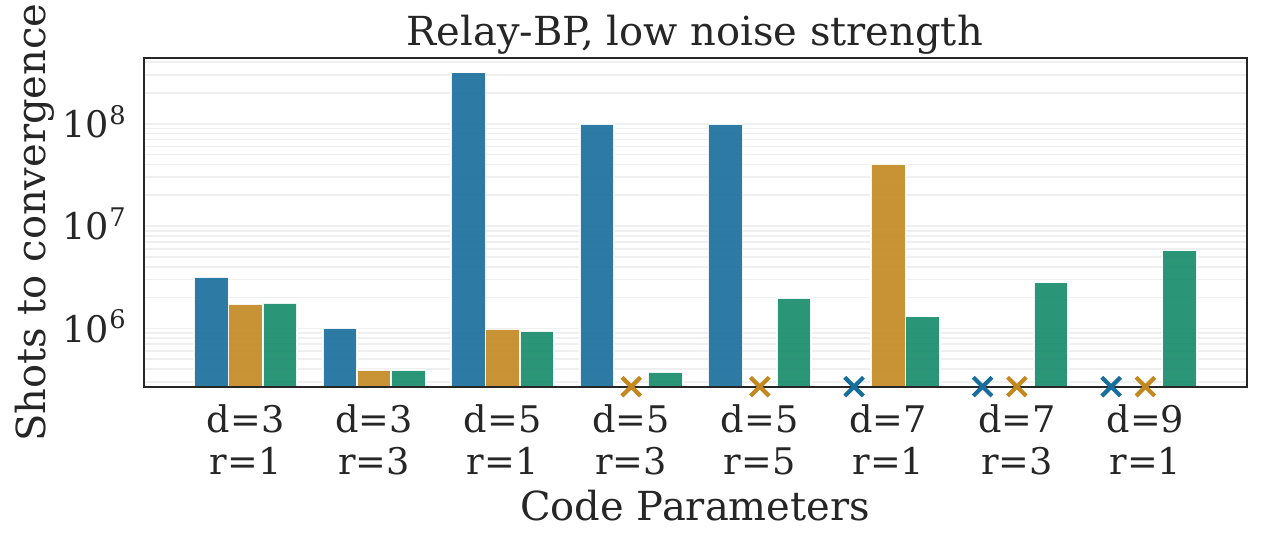}
    \caption{$p = 0.0001$}
    \label{fig:conv-rate-low-noise}
  \end{subfigure}
  \caption{Convergence for the logical error rate problem.}
  \label{fig:conv-rate}
\end{figure}

%% file: fig-bp-osd-vary-noise-strength.tex
%!TEX root=./main.tex
\begin{figure}
  \begin{subfigure}[b]{0.29\linewidth}
    \includegraphics[width=\textwidth]{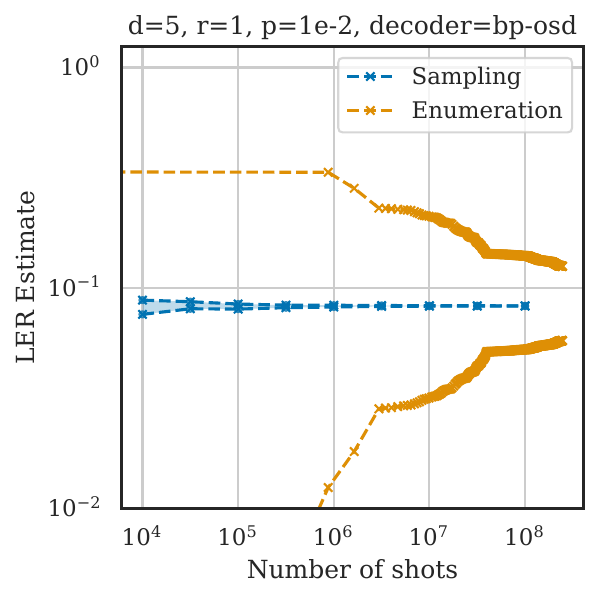}
    \caption{p = 0.01}
  \end{subfigure}
    \begin{subfigure}[b]{0.29\linewidth}
    \includegraphics[width=\textwidth]{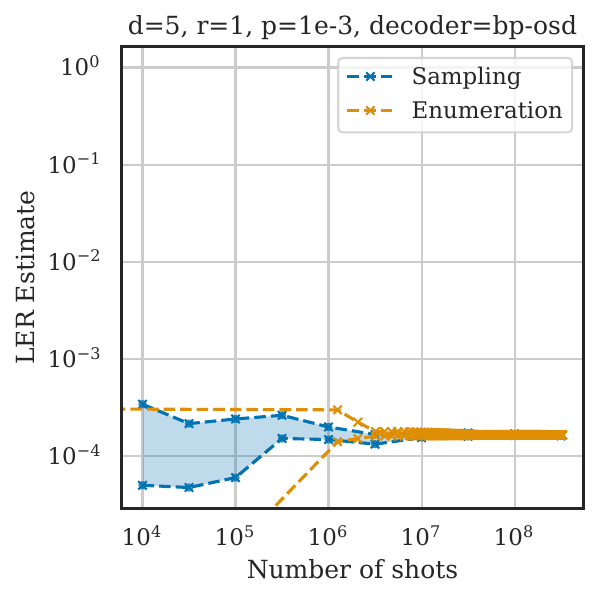}
        \caption{p = 0.001}
  \end{subfigure}
    \begin{subfigure}[b]{0.29\linewidth}
    \includegraphics[width=\textwidth]{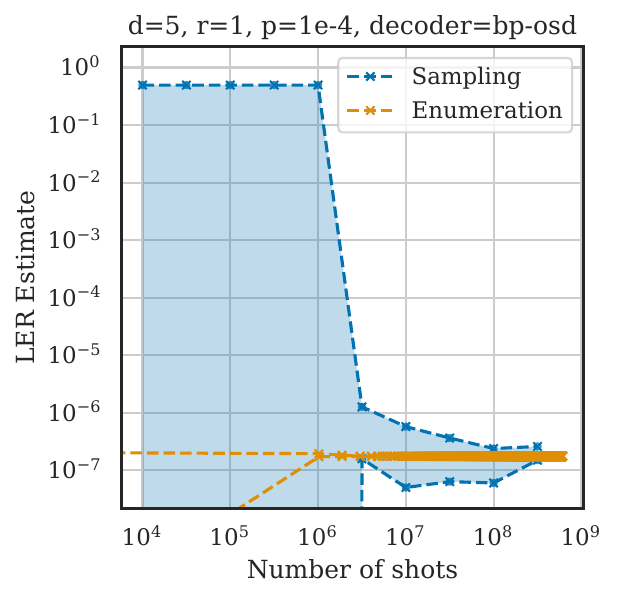}
    \caption{p = 0.0001}
  \end{subfigure}
  \caption{Comparing Enumeration to Sampling. Each plot shows estimation quality as a function of shots at different noise strengths for the same program.}
  \label{fig:bp-osd-vary-noise-strength}
\end{figure}

%% file: fig-pymatch-d7-r1.tex
%!TEX root=./main.tex
\begin{figure}
  \begin{subfigure}[b]{0.29\linewidth}
    \includegraphics[width=\textwidth]{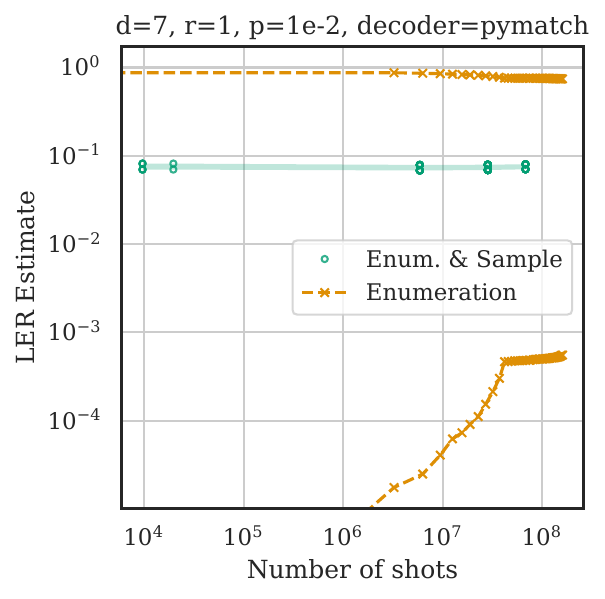}
    \caption{p = 0.01}
  \end{subfigure}
    \begin{subfigure}[b]{0.29\linewidth}
    \includegraphics[width=\textwidth]{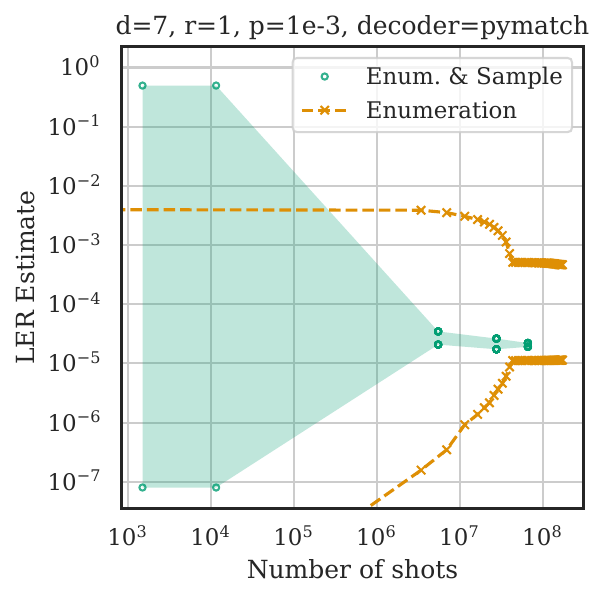}
        \caption{p = 0.001}
  \end{subfigure}
    \begin{subfigure}[b]{0.29\linewidth}
    \includegraphics[width=\textwidth]{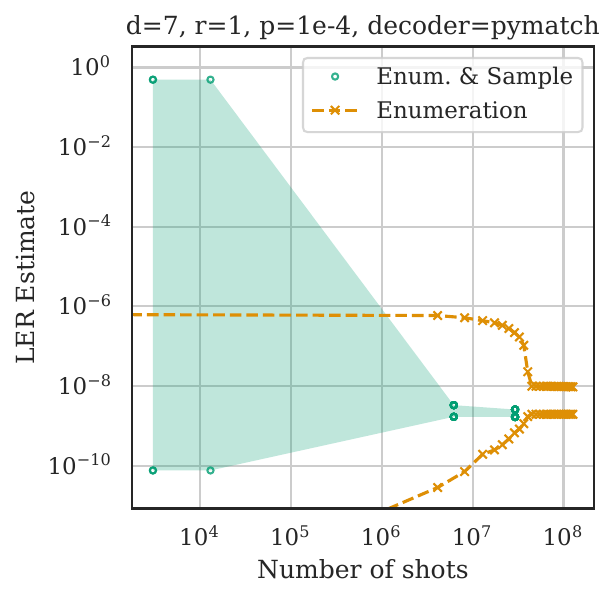}
    \caption{p = 0.0001}
  \end{subfigure}
  \caption{Comparing Enumeration alone to Enumeration \& Sampling. Each plot shows estimation quality as a function of shots at a different noise strength for the same program.}
  \label{fig:pymatch-d7-r1}
\end{figure}

%% file: fig-rob-conv-rate.tex
%!TEX root=./main.tex
\begin{figure}
  \includegraphics[width=0.85\textwidth]{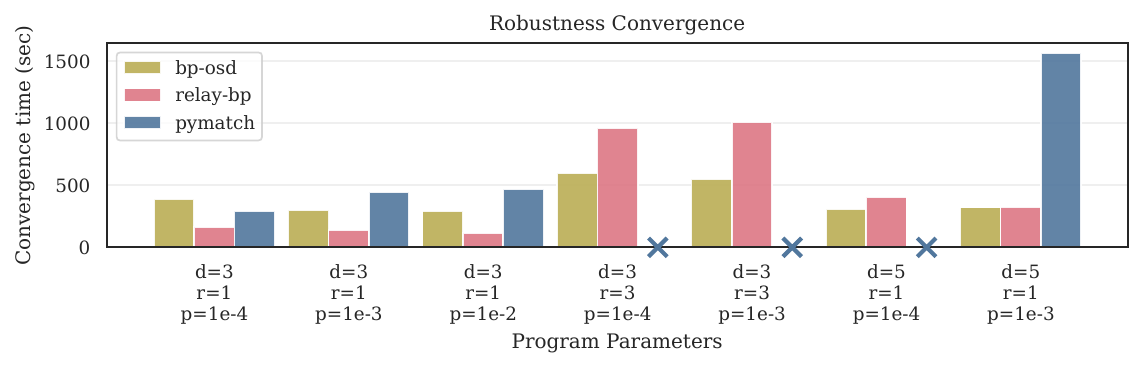}
  \caption{Convergence for the Robustness Problem.}
  \label{fig:rob-conv-rate}
\end{figure}

%% file: fig-split-search.tex
%!TEX root=./main.tex
\begin{figure}
  \begin{subfigure}[b]{0.29\linewidth}
    \includegraphics[width=\textwidth]{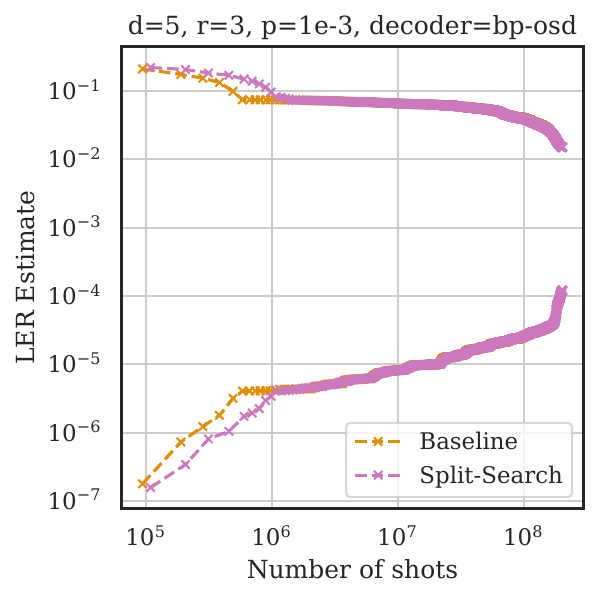}
    \caption{ d = 5 }
  \end{subfigure}
    \begin{subfigure}[b]{0.29\linewidth}
    \includegraphics[width=\textwidth]{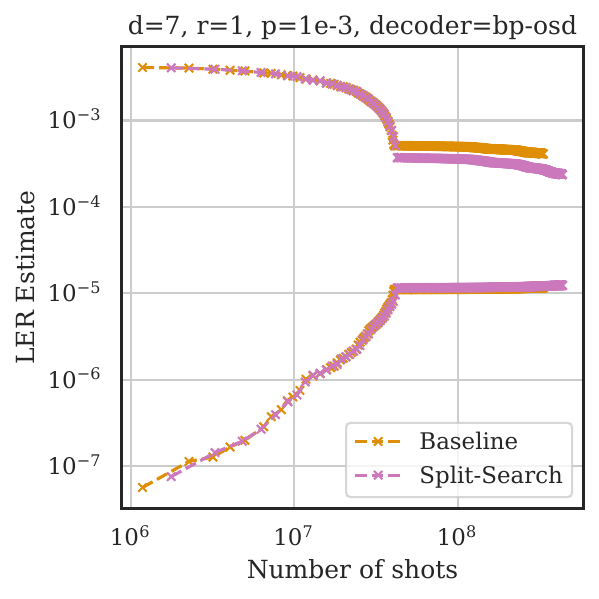}
        \caption{d = 7}
  \end{subfigure}
    \begin{subfigure}[b]{0.29\linewidth}
    \includegraphics[width=\textwidth]{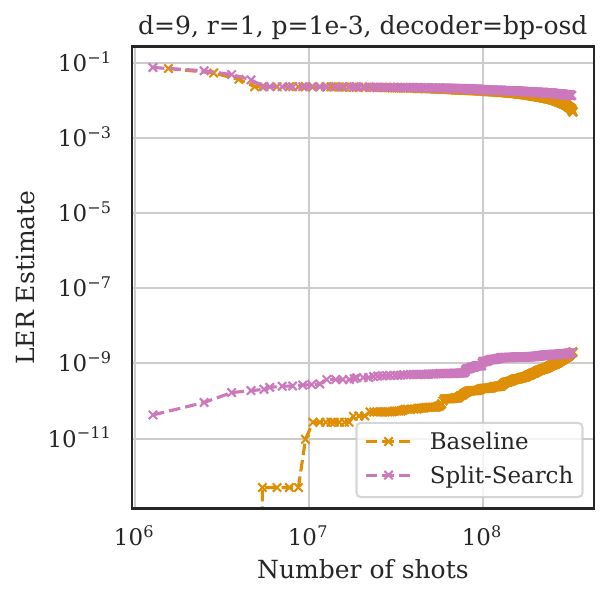}
    \caption{d = 9}
  \end{subfigure}
  \caption{Evaluating the split-search strategy. Each plot shows a program with a different code distance.}
  \label{fig:split-search}
\end{figure}

%% file: fig-ls-tables.tex
%!TEX root=./main.tex
\begin{figure}
  \footnotesize
\centering
\begin{tabular}{l cc}
\toprule
 & Frac.\ improved & Geo. Mean Shot Ratio \\
\midrule
\multicolumn{3}{l}{\itshape Enumeration } \\
\quad \textsc{flip} & 13/25 (0.52) & \textcolor{red}{1.05} \\
\quad \textsc{shift} & 24/25 (0.96) & \textcolor{green!60!black}{0.73} \\
\quad \textsc{shift} $\cup$ \textsc{flip} & 18/25 (0.72) & \textcolor{green!60!black}{0.84} \\
\midrule
\multicolumn{3}{l}{\itshape Enumeration \& Sampling} \\
\quad \textsc{flip} & 46/81 (0.57) & \textcolor{red}{1.08} \\
\quad \textsc{shift} & 55/81 (0.68) & \textcolor{green!60!black}{0.89} \\
\quad \textsc{shift} $\cup$ \textsc{flip} & 53/81 (0.65) & \textcolor{green!60!black}{0.97} \\
\bottomrule
\end{tabular}
\caption{Evaluating local search variants as compared to a Hamming weight order baseline.
The shot ratio for a benchmark is the number of shots until convergence for the local move strategy divided by that of the baseline search. 
}
\label{fig:ls-tables}
\end{figure}

%% file: appendix.tex
%!TEX root=./main.tex
\newpage
\appendix
\section{Proofs}\label{appendix:start}
\paragraph{\zcref{thm:ler-is-poly-eval}} The proof of this theorem is notationally dense, but conceptually a straightforward formalization of the reasoning developed in \zcref{sec:prob-poly}.
We start with the probability of a single error bitstring. 
\begin{lemma}
    For a \qec program $\prog(\vctr{v})$ and error bitstring $e$, we have that $f(e; \prog, \drawseqset) = m_e(\vctr{v})$
\end{lemma}
\begin{proof}
   First, recall the random draw space $\delta$ is chosen precisely so that it can model the probabilities of individual Bernoulli error channels in the program $\prog(v_1, \ldots, v_i)$, meaning that for each argument $v_i$:
\[ 
\frac{|\sigma \in \delta  : E_i(\sigma_i) \neq  I| }{|\delta|} = v_i.
\]
It follows that 
\[ 
 \frac{|\sigma \in \delta  : E_i(\sigma_i) =  I| }{|\delta|} = 1 - v_i.
\]
Next, note that we can decompose the probability of an error bitstring in the probabilities of each individual error channel.
\[ 
f(e; \prog, \drawseqset) = \frac{|\{\Sigma \in \drawseqset  : e_\Sigma = e\}|}{|\drawseqset|} = \prod_{e_i = 1} \frac{|\sigma \in \delta  : E_i(\sigma_i) \neq  I| }{|\delta|} \prod_{e_i = 0} \frac{|\sigma \in \delta  : E_i(\sigma_i) =  I| }{|\delta|}.
\]
Here, $\drawseqset$, as before, refers to the space of sequences of $n$ draws, $\drawseqset = \delta^n.$
Combining these facts, we obtain 
\[f(e; \prog, \drawseqset) = \prod_{e_i = 1} v_i \prod_{e_i = 0}(1-v_i).\]

But this is exactly the evaluation of the error minterm for $e$ at the point $\vctr{v}$, completing the proof.
\end{proof}
With this lemma in hand, we can proceed to prove the theorem. 
Starting with \zcref{def:ler}, we rewrite the Accuracy as a probability over error bitstrings
\[\err(\prog, \decoder) = \textrm{Pr}_{(s,o) \sim P}[\decoder(s) \neq o] = \textrm{Pr}_{e \sim P}[\decoder(\syndfunc(e)) \neq \obsfunc(e)] = \sum_{e \in \errorset} f(e; \prog, \drawseqset) \] 
Then, we apply the lemma to reach the equality
\[\err(\prog, \decoder) = \sum_{e\in\errorset} m_e(\vctr{v}).\]
The righthand side here is the definition of $p_L(\vctr{v})$, so we are done. 
\paragraph{\zcref{thm-rob-is-poly-opt}} This follows directly from \zcref{thm:ler-is-poly-eval}.

\paragraph{\zcref{thm:bounds}} Observe that error polynomials are positive over this space (consistent with their interpretation as a probability).
Thus, we obtain the first inequality by noting that $p_{L} = p_{L \cap S} + p_{L \cap S^c} \geq p_{L \cap S}$.
For the second, we additionally use the fact that $p_L = 1-p_{L^c}$ and then apply the same reasoning:  $p_{L^c} = p_{S \setminus L} + p_{S^c \setminus L} \geq  p_{S \setminus L}$,
meaning that $p_L = 1 - p_{L^c} \leq 1- p_{S \setminus L}$. 

\paragraph{\zcref{thm:rob-bounds}} Follows from \zcref{thm:bounds} and the simple observation that $\max_{x \in X}(1-x) = 1 - \min_{x \in X}(x)$.

\paragraph{\zcref{thm:opt-at-corner}} Let $r$ be an interior point of the hyperrectangle, so that there exists some $i$ such that $r_i \not \in \{\ell_i, u_i\}.$
Consider the partial derivative $\frac{\partial p}{\partial x_i} = a(1-x_i) + b(x_i)$ at $r$. 
If $a > b$, then let $r'$ be equal to $r$ except that $r_i = \ell_i$ and note that $p(r') > p(r).$ Otherwise, choose $r'_i = u_i$, and we have that $p(r') \geq p(r).$
Repeated application of this argument shows that any value of $p$ obtained by an interior point is at least matched at a corner. 

\paragraph{\zcref{thm:alg-correct}}
Let $v$ be the vertex obtained by this procedure and $v^*$ be a vertex that achieves a global optimum. 
First, note that $v^*$ cannot differ from $v$ at an index $i$ for which we certify positive or negative slope over the entire set. 
If it did, then in particular $\frac{\partial p}{\partial x_i}(v^*)$ would have the corresponding sign and we could choose the other value for $x_i$ and increase the value of $p$. 
Then, since we exhaustively search over remaining variables, the one chosen by $v$ must be at least as good as any choice, including the one of $v^*$.

\paragraph{\zcref{thm:conf}}
We start with the Kullback-Leibler formulation of the Chernoff bounds:
\begin{align}
     & \Pr_\theta[\hat{\Theta} \geq u] \leq e^{-N\cdot KL(u||\theta)}  && (u \in [\theta,1]) \label{eq:chernoff-hi} \\
     & \Pr_\theta[\hat{\Theta} \leq u] \leq e^{-N\cdot KL(u||\theta)}  && (u \in [0,\theta]) \label{eq:chernoff-lo}\\
     \text{where } & KL(p || q) = p\ln\left(\frac{p}{q}\right) + (1-p)\ln\left(\frac{1-p}{1-q}\right). \notag
\end{align}

The goal is to invert this bound to construct a $1-\alpha$ confidence interval $[L,U]$ based on the point estimator $\hat{\Theta}$ for $\theta$.

Observe that, for fixed $\theta$, the function $u \mapsto KL(u || \theta)$  is strictly decreasing on the interval $(0,\theta]$, is strictly increasing on $[\theta, 1)$, and is 0 at $\theta$.
Therefore, for any $t > 0$, there exist unique values $u_t^- \in [0, \theta]$ and $u_t^+ \in [\theta, 1]$ such that 
\begin{align}
    N \cdot KL(u_t^-||\theta) = t, \qquad N \cdot KL(u_t^+||\theta) = t.
    \label{eq:kl-equal-t}
\end{align}

Further, the monotonicity of this function implies that
\begin{equation}
    \sett*{N \cdot KL(\hat{\Theta}||\theta) \geq t} \subseteq \sett*{\hat{\Theta} \leq u_t^-}  \cup \sett*{\hat{\Theta} \geq u_t^+}.
    \label{eq:kl-set-contain}
\end{equation}

Therefore
\begin{align*}
\Pr_\theta\left[ N \cdot KL(\hat{\Theta}||\theta) \geq t\right]
    &\leq \Pr_\theta\left[\sett*{\hat{\Theta} \leq u_t^-}  \cup \sett*{\hat{\Theta} \geq u_t^+}\right] && \mbox{(\zcref{eq:kl-set-contain} and monotonicity of measure)} \\
    &\leq \Pr_\theta\left[\hat{\Theta} \leq u_t^-\right] + \Pr\left[\hat{\Theta} \geq u_t^+\right] && \mbox{(union bound)}\\
    &\leq e^{-N \cdot KL(u_t^-||\theta)} + e^{-N \cdot KL(u_t^+||\theta)}  && \mbox{(\zcref{eq:chernoff-hi,eq:chernoff-lo})} \\
    &= 2e^{-t} && \mbox{(\zcref{eq:kl-equal-t}).}
\end{align*}
Choose $t_\alpha \coloneq \ln(2/\alpha)$ so that $2e^{-t_\alpha} = \alpha$.
Then
\begin{align*}
\Pr_\theta\left[ N \cdot KL(\hat{\Theta}||\theta) \leq t_\alpha\right] \geq 1 - \alpha.
\end{align*}
This configuration gives a $1-\alpha$ confidence interval $[L,U]$ where
\begin{align*}
L &\coloneq \inf\sett*{u \in [0, \hat{\Theta}] : KL(\hat{\Theta}|| u) \leq \frac1N(\ln(2/\alpha))}, \\
U &\coloneq \sup\sett*{u \in [\hat{\Theta}, 1] : KL(\hat{\Theta}|| u) \leq \frac1N(\ln(2/\alpha))}.
\end{align*}

On the event $0 < \hat{\Theta} < 1$, the monotonicity of the $u \mapsto KL(\hat\Theta || u)$ function ensures that $L$ and $U$ are the unique values which achieve equality.
If $\hat\Theta = 0$ then $L = 0$ and
\begin{equation}
KL(0||u) = \ln(1/(1-u)) = 1/N\ln(2/\alpha) \implies U = 1 - (\alpha/2)^{1/N}.
\end{equation}
If $\hat\Theta = 1$ then $U = 1$ and
\begin{equation}
KL(1||u) = \ln(1/u) = 1/N\ln(2/\alpha) \implies L = (\alpha/2)^{1/N}.
\end{equation}